\documentclass[lettersize,journal]{IEEEtran}
\usepackage{amsmath,amsfonts, amsthm}
\usepackage{algorithmic}
\usepackage{array}
\usepackage[caption=false,font=normalsize,labelfont=sf,textfont=sf]{subfig}
\usepackage{textcomp}
\usepackage{stfloats}
\usepackage{url}
\usepackage{verbatim}
\usepackage{graphicx}
\usepackage{balance}
\usepackage[dvipsnames, table]{xcolor}
\usepackage{xargs}

\usepackage{tabularx}
\usepackage{siunitx}
\usepackage{multirow}
\usepackage{array, diagbox}
\usepackage[utf8]{inputenc}
\usepackage{booktabs}      

\newcolumntype{Y}{>{\centering\arraybackslash}X}

\usepackage{amsmath, amssymb, amsfonts}
    
\newcommand{\average}[1]{\overline{#1}}
\newcommand{\avgest}[1]{\estimate{\average{#1}}}

\newcommand{\altparameter}{\vartheta}
\newcommand{\classK}{\mathcal{K}}

\newcommand{\continuous}{\set{C}}
\newcommand{\costfunction}{{J}}
\newcommand{\differentialoperator}[1]{D^{#1}}
\newcommand{\derivativevalue}{\xi}
\DeclareMathOperator{\diagonal}{diag}
\newcommand{\ditheramplitude}{a}
\newcommand{\dithersignal}{p}
    \newcommand{\extendeddithersignal}{\rho}
\newcommand{\ditherfrequency}{\omega}
\newcommand{\error}[1]{\widetilde{#1}}
\newcommand{\estimate}[1]{\hat{#1}}

\newcommand{\gradient}{{g}}

\newcommand{\hessian}{{H}}

\newcommand{\integernumbers}{\mathbb{Z}}
    \newcommand{\nonnegativeintegernumbers}{\integernumbers_{\geq 0}}
\newcommand{\invhessian}{{\Gamma}}

\newcommand{\multiindex}{\alpha}
    \newcommand{\multiindexsize}{c}

\newcommand{\parameter}{\theta}

\newcommand{\realnumbers}{\mathbb{R}}

\newcommand{\rmd}{{d}}
\newcommand{\set}[1]{\mathcal{#1}}

\DeclareMathOperator{\trace}{tr}

\newcommand{\ball}[2][]{\set{B}^{{#1}}_{{#2}}}
\newcommandx{\derivative}[3][1={}, 2={}]{\frac{{\rmd^{#2} #1}}{{\rmd #3^{#2}}}}
    \newcommand{\derivativefilter}{h}

\newcommand{\angularvelocity}{\omega}
\newcommand{\commanded}[1]{{#1}_{com}}
\newcommand{\forwardvelocity}{v}
\newcommand{\headingangle}{\theta}
    \newcommand{\sensorangletosource}{\beta}

\newcommand{\pressure}{p}
\newcommand{\reference}[1]{{#1}_{\mathrm{ref}}}

\newcommand{\demodulationsignal}{h}
\newcommand{\microphoneangle}{\phi}

\newcommand{\state}{x}
    \newcommand{\altstate}{z}

\newcommand{\gain}[1]{k_{#1}}

\newcommand{\period}[1]{#1}
\newcommand{\perturbation}[1]{{#1}}
\newcommand{\covariance}{Q}
    \newcommand{\crossvariance}{R}

\DeclareMathOperator{\suchthat}{s.t.}

\DeclareMathOperator{\halfvect}{vech}

    \newtheorem{thm}{Theorem}
    
    \newtheorem{prop}[thm]{Proposition}
    \newtheorem{asmp}{Assumption}
    
    \newtheorem{rem}{Remark}
    
\usepackage{nicematrix}
\usepackage{tikz}
    \usetikzlibrary{
        decorations.pathreplacing,
        calc,
        arrows,
        shapes,
        shapes.symbols,
        automata,
        calligraphy, 
        positioning, 
        backgrounds, 
        fit, 
        arrows.meta
    }

\tikzset{%
    block/.style = {
        draw, thick, rectangle, minimum height=3em, minimum width=3em
    },
    sum/.style = {draw, circle,minimum size=1cm}, 
    convolve/.style = {
        inner sep=0,
        outer xsep=-1pt,
        outer ysep=0,
        scale=2
    },
    input/.style    = {coordinate}, 
    output/.style   = {coordinate}, 
    gain/.style = {
      	draw, 
        isosceles triangle,
        isosceles triangle apex angle=50,
        minimum height = 3.0em,
        outer sep=0
    },
}

\begin{document}
\title{
    Servos for Local Map Exploration Onboard Nonholonomic Vehicles for Extremum Seeking
}

\author{Dylan James-Kavanaugh$^{1,\ast}$, Patrick McNamee$^{1,\ast}$, Qixu Wang$^{1}$, and Zahra Nili Ahmadabadi$^{1}$

\thanks{
    Research was sponsored by the Army Research Office under Grant
    Number W911NF-24-1-0386 and the Department of the Navy, Office of Naval Research
    under ONR award number N000142412269. The views, findings, conclusions, or recommendations contained in this document are those of the authors and should not be interpreted as representing the official policies or views, either expressed or implied, of the Army Research Office, the Office of Naval Research, or the U.S. Government. The U.S. Government is authorized to reproduce and distribute reprints for Government purposes notwithstanding any copyright notation herein.
}
    \thanks{
        $^{1}$ D. James-Kavanaugh, P. McNamee, Q. Wang, and Z. N. Ahmadabadi are with the Department of Mechanical Engineering, San Diego State University, San Diego, CA, USA
                \{{\tt\small djameskavanaug2312, \tt\small pmcnamee5123}, {\tt\small qwang0429}, and {\tt\small zniliahmadabadi}\} {\tt\small @sdsu.edu }}%
    \thanks{
        $^{\ast}$ These authors contributed equally as joint first authors.
    }
}

\markboth{Journal of \LaTeX\ Class Files,~Vol.~18, No.~9, September~2025}%
{How to Use the IEEEtran \LaTeX \ Templates}

\maketitle

\begin{abstract}
    Extremum seeking control (ESC) often employs perturbation-based estimates of derivatives for some sensor field or cost function. These estimates are generally obtained by simply multiplying the output of a single-unit sensor by some time-varying function. Previous work has focused on sinusoidal perturbations to generate derivative estimates with results for arbitrary order derivatives of scalar maps or higher up to third-order derivatives of multivariable maps.  This work extends the perturbations from sinusoidal to bounded periodic or almost periodic functions and considers multivariable maps. A necessary and sufficient condition is given for determining if time-varying functions exist for estimating arbitrary order derivatives of multivariable maps for any given bounded periodic or almost periodic dither signal. These results are then used in a source seeking controller for a nonholonomic vehicle with a sensor actuated by servo. The conducted simulation and real-world experiments demonstrate that by distributing the local map exploration to a servo, the nonholonomic vehicle was able to achieve a faster convergence to the source.
\end{abstract}


\section{Introduction}
    \label{sec:introduction}
    
Extremum seeking control (ESC) is a family of continuous-time optimization algorithms which drive estimates of optimal parameters for some map to an extremum. Most ESCs will act in a model-free way, that is to say, they do not have access to the derivatives but rather estimate them by multiplying a sensor output by a time-varying function. One of the first ESCs can be traced back to Leblanc in the 1920s \cite{ref:leblanc-1922}, where it was introduced in the context of automatic tuning of electrical machinery. The technique gained prominence in the 1950s and 1960s \cite{ref:tan-2010}, with numerous studies exploring adaptive search methods to locate optimal operating points in aerospace and industrial processes. Since then, ESCs have undergone a resurgence when the stability of ESC was rigorously proven with a feedback scheme \cite{ref:krstic-2000} and pushed the interest in the direction it is today. Contemporary research shows its applicability in a wide range of fields, such as automotive applications \cite{ref:drakunov-2002} and autonomous robot navigation \cite{ref:zhang-2007, ref:durr-2017}. The design of ESCs remains similar in previous literature, perturbations are added to the estimates of optimal parameters, and these estimates are updated based on the output of the perturbed map \cite{ref:nesic-2006, ref:nesic-2010-2}. By adjusting the optimal parameter estimates based on these estimated derivatives, the controller achieves a practical stability to an extremum. For nonholonomic vehicles, the perturbation is generated by the vehicle's motion itself, i.e.,\ the vehicle is responsible for both tasks of local exploration as well as adjusting the optimal parameters.

Sinusoidal perturbation signals have dominated the previous literature. Ne{\v s}i{\'c} \emph{et al.} \cite{ref:nesic-2010-2} was the first to give a generalized framework for estimating derivatives of arbitrary order based on sinusoidal perturbation signals for scalar maps. References \cite{ref:mills-2018} and \cite{ref:rusiti-2016} gave explicit perturbation signals for higher order derivatives of a scalar map using sinusoidal signals, although we show in Appendix that these results can actually be derived from \cite{ref:nesic-2010-2}. In contrast to the scalar case, multivariable maps have received comparatively less attention, although \cite{ref:ghaffari-2012} showed the rules required for sinusoidal perturbations to estimate second-order derivatives of multivariable maps and later \cite{ref:ghaffari-2024} gave the rules for estimating the third-order derivatives. These rules are based on the relative dither rates of the individual element perturbation signals comprising the overall vector perturbation signals, and the number of rules grow fast for higher and higher order derivatives required. For instance, only one rule is required for estimating the first-order derivatives, 5 rules for the second-order derivatives \cite{ref:ghaffari-2012}, and 29 rules for the third-order derivatives \cite{ref:ghaffari-2024}. 

Perturbation signals other than the sinusoidal signals has been in previous literature. Tan \emph{et al.} \cite{ref:tan-2007} has shown that the choice of the perturbation signal, including non-sinusoidal signals, directly influences the performance of gradient-based extremum seeking control (GESC) on scalar maps. Additionally, \cite{ref:haring-2023}  used non-sinusoidal perturbation signals, all periodic with a common period, for gradient-based approaches on multivariable maps. Our work extends these findings by considering any periodic or almost-periodic signals for the purpose of estimating derivatives of arbitrary order in multivariate maps. We show a necessary and sufficient condition for determining if a periodic or almost-periodic perturbation signal vector has time-varying functions which can be used to estimate the derivatives of a multivariable map, even if the element-wise signals themselves are not orthogonal. Additionally, we demonstrate the importance of our results by utilizing it for implementing the novel concept of using a servo to perturb a sensor onboard a novel nonholonomic vehicle for local map exploration for derivative estimations. In doing so, we show that this implementation outperforms the state-of-the-art source seeking in both simulations and real-world experiments.
 
\textit{Main Contribution:} 

This work builds on the existing research of \cite{ref:nesic-2010-2, ref:mills-2018, ref:ghaffari-2012} and \cite{ref:tan-2007} to develop a comprehensive method to determine whether perturbation-based estimation of arbitrary $m$-th ordered derivatives is possible when using specified periodic or almost-periodic dither signals. We give a necessary and sufficient condition for the existence of time-varying functions which can be multiplied by the sensor output to achieve the arbitrary order derivative estimates. These signals allow for the average derivative estimates to converge to the true derivative values for all cost functions at least $m$-times differentiable. Additionally, we demonstrate that, when utilizing these results, one can design a controller for a nonholonomic vehicle which has better convergence times than controllers from previous literature where the vehicle was simultaneously in charge of local exploration and moving towards the extremum.

\textit{Notation:}

Multi-index notation is extensively used in this work. A multi-index $\multiindex$ is an $n$-tuple of nonnegative integer entries ($\multiindex\in\nonnegativeintegernumbers^n)$. The sum of $\multiindex$ is $\vert \multiindex \vert  = \sum_{i=1}^n \multiindex_i$ and the factorial is defined similarly as $\multiindex! = \prod_{i=1}^n \multiindex_i!$. A vector $x\in\realnumbers^n$ can be given as a power $x^\multiindex = \prod_{i=1}^n {x_i}^{\multiindex_i}$. Additionally, multi-index can be used to compactly express the $m$-th order Taylor series expansion of a function $f:\realnumbers^n\to\realnumbers$ at the point $x_0\in\realnumbers^n$ as
\begin{equation}
    f(x) = \sum_{\vert\multiindex\vert = 0}^{m} \frac{1}{\multiindex!} \differentialoperator{\multiindex}f(x_0)(x - x_0)^{\multiindex}
\end{equation}
where $\differentialoperator{\multiindex} = \left(\frac{\partial}{\partial x_1}\right)^{\multiindex_1}\cdots \left(\frac{\partial}{\partial x_n}\right)^{\multiindex_n}$. We assume that there is a set $\set{A}$ that contains all $\multiindex$ that we will consider for the ESCs.  For brevity, we will assume that $\set{A}$ has some implicit ordering so that, from now on, $\multiindex_i$ is the $i$-th element of $\set{A}$ rather than the $i$-th element of the $n$-tuple.

\section{Problem Statement}
    \label{sec:problem-statement}

The general goal of ESCs is to find the extrema of a cost function $\costfunction(\parameter)$. Traditionally, ESCs expand $\parameter$ into an additive combination of the estimate of the optimal parameter $\estimate{\parameter}$ and a local perturbation/dither signal $\dithersignal(\ditherfrequency t)$, i.e. $\parameter = \estimate{\parameter} + \ditheramplitude  \dithersignal(\ditherfrequency t)$. The small perturbations ($\ditheramplitude \ll 1$) occur on a fast time scale ($\ditherfrequency \gg 1$) while the estimated optimal parameter is updated on a slow time scale by differential equations
\begin{equation}
    \label{eq:problem-statement:original-x-ode}
    \derivative{t}\estimate{\state} = f(\omega t, \estimate{\state}, \ditheramplitude)
\end{equation}
where $\estimate{\state} \in \realnumbers^{n+n'}$ is the augmented state vector containing $n$ parameters and $n'$ auxiliary states which often estimate the derivative information. Typically, these systems are affine with respect to the derivative estimates, so we can use the affine dynamics of
\begin{equation}
    \label{eq:problem-statement:parameter-update-ode:affine}
    \derivative{t} \estimate{\state} = f_0\left(\estimate{\state}\right) + f_1\left(\estimate{\state}\right) \estimate{\derivativevalue}\left(\omega t, \estimate{\parameter}, \ditheramplitude\right)
\end{equation}
where $\estimate{\derivativevalue}\in\realnumbers^{p}$ is a vector of the estimates of the derivatives $\derivativevalue_i(\estimate{\parameter}) = \differentialoperator{\multiindex_i}\costfunction\left(\estimate{\parameter}\right)$. To better illustrate how often ESCs with affine dynamics appear in the literature, we give some examples of ESCs with differential equations in this form.

\subsection{Examples of Affine Extremum Seekers}

\subsubsection{Gradient Descent \cite{ref:zhang-2007}} 
The Gradient-based Extremum Seeking Control is the baseline ESC design. It relies on estimating the gradient of a scalar objective function with respect to tunable parameters through periodic perturbations and demodulation. The parameter estimate update dynamics are given by
\begin{equation}
    \derivative{t} \estimate{\parameter} = - \gain{} \estimate{\gradient}(\ditherfrequency t, \estimate{\parameter}, \ditheramplitude),
\end{equation}
where $\estimate{\state}=\estimate{\parameter}$ as there are no auxiliary states and $\estimate{\gradient} = \estimate{\derivativevalue}$ is the perturbation-based estimate of the gradient $\nabla\costfunction$. This system can be expressed in the form of \eqref{eq:problem-statement:parameter-update-ode:affine} with $f_0(\estimate{\parameter})~=~0$ and $ f_1(\estimate{\parameter}) = -\gain{}$.

\subsubsection{Heavy-ball \cite{ref:michalowsky-2014}} 
The Heavy-ball ESC has the dynamics of
\begin{align}
    \derivative{t} \estimate{\parameter} & {} = {} \phi \\
    \derivative{t} \phi & {} = {} -\beta \phi - \gain{} \estimate{\gradient}(\ditherfrequency t, \estimate{\parameter}, \ditheramplitude)
\end{align}
where the state vector $\estimate{x} = [\estimate{\parameter}^T, \phi^T]^T$ is augmented with the auxiliary state $\phi \in \realnumbers^n$ and $\beta \in \realnumbers$ is the dampening term. This system can be rewritten with affine dynamics using $f_0(\estimate{x})~=~[\phi^T,~-\beta \phi^T]^T$ and $f_1(\estimate{x}) = [0,\ -\gain{}I_{n\times n}]^T$ where $I_{n\times n}$ is the $n\times n$ identity matrix.

\subsubsection{Newton Method \cite{ref:ghaffari-2012,ref:mcnamee-2025:scalar-newton-automatica,ref:mcnamee-2024:scalar-maps}} 
The Newton-based ESC (NESC) corrects for second-order derivatives $\nabla^2\costfunction$ by estimating the inverse of the Hessian. The dynamical system is given in \cite{ref:ghaffari-2012} as
\begin{align}
    \derivative{t}\estimate{\parameter} &= -\gain{}\estimate{\invhessian}\estimate{\gradient}\left(\ditherfrequency t, \estimate{\parameter}, \ditheramplitude\right) \\
    \derivative{t}\estimate{\invhessian} &= \omega_{\ell}\left(\estimate{\invhessian} - \estimate{\invhessian}\, \estimate{\hessian}(\ditherfrequency t, \estimate{\parameter}, \ditheramplitude)\,\estimate{\invhessian}\right)
\end{align}
where $\estimate{\invhessian}$ is the estimate of the inverse of the Hessian and $\estimate{\hessian}$ is the perturbation-based estimate of the Hessian. As it stands, the NESC system is a combination of a vector differential equation and a matrix differential equation, so its initially unclear how this can be expressed as an affine vector differential equation. However, the system can be expressed as
\begin{multline}
    \derivative{t} \begin{bmatrix}
        \estimate{\parameter}\\
        \halfvect (\estimate{\invhessian})
    \end{bmatrix} = \begin{bmatrix}
        0\\
        \omega_{\ell} \halfvect (\estimate{\invhessian}) 
    \end{bmatrix} +\\ \begin{bmatrix}
            -k \estimate{\invhessian} & 0\\ 0 & 
            \omega_{\ell}L_n(\estimate{\invhessian} \otimes \estimate{\invhessian}) D_n
        \end{bmatrix} \begin{bmatrix}
            \estimate{\gradient}(\ditherfrequency t, \estimate{\parameter}, \ditheramplitude)\\
            \halfvect (\estimate{\hessian}(\ditherfrequency t, \estimate{\parameter}, \ditheramplitude))
        \end{bmatrix},
\end{multline}
where $\otimes$ is the Kronecker product, $\halfvect$ is the half-vectorization operator, and the matrices $L_n$ and $D_n$ are elimination and duplication matrices to account for the symmetry of the symmetric matrices. In this form, the state vector of the NESC exists in $\estimate{\state}\in\realnumbers^{n + n(n+1)/2}$. 

\subsection{Averaging of Affine Extremum Seekers}

\begin{figure}[t!]
    \centering
    \resizebox{0.45\textwidth}{!}{
\begin{tikzpicture}[auto, node distance=2cm,>=Latex]
    \node[text width=13ex,align=center] at (0,0)[block] (cost) {Nonlinear $\costfunction(\parameter)$};
    \node [input] (input) {};
    \node at (-4,0) (parameters) {};
    \node [output] (output) {};
    \node at (3,0) (y-output) {};

    \node at (2,-2) [convolve] (convolution) {$\otimes$};

    \node at (0,-2) (demodulation) {};
    \node at (0,-4) (gradient-estimate) {};
    
    \draw[->] (parameters) -- node {$\parameter =\estimate{\parameter} + \ditheramplitude \perturbation{p}(\ditherfrequency t)$}(cost);
    \draw[->] (cost.east) --  node[pos=0.9,above] {$y$} (y-output);
    \draw[->] (cost.east) -| (convolution.north);
    \draw[->] (demodulation) --  node[pos=0.1,left] {$\demodulationsignal(\ditherfrequency t,\ditheramplitude)=\begin{bmatrix}
        \demodulationsignal_1(\ditherfrequency t,\ditheramplitude)\\
        \vdots\\
        \demodulationsignal_o(\ditherfrequency t,\ditheramplitude)
    \end{bmatrix}$} (convolution.west);
    \draw[->] (convolution) |- node [pos=1.0,left] {$\begin{bmatrix}
        \estimate{\xi}_1(\ditherfrequency t,\estimate{\parameter},\ditheramplitude)\\
        \vdots\\
        \estimate{\xi}_o(\ditherfrequency t,\estimate{\parameter},\ditheramplitude)
    \end{bmatrix}$}(gradient-estimate);
\end{tikzpicture}
    }
    \caption{Block diagram of derivative estimator}
    \label{blk:problem-statement:estimation-block-diagram}
\end{figure}
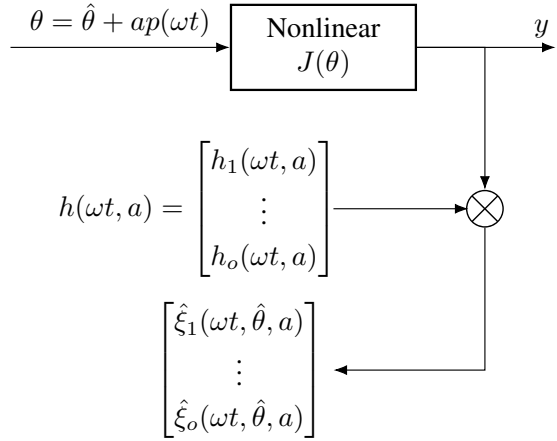

While affine systems are prevalent for ESCs, one still needs to average the dynamics to obtain the average system for stability analysis. The following assumptions are necessary for an $m$-th order ESC. 
\begin{asmp}
    The cost function $\costfunction: \realnumbers^n \to \realnumbers$ is $m$-times differentiable ($\costfunction\in \continuous^m$) for the $m$-th order algorithm.
    \label{asmp:problem-statement:assumption-1-cost-function}
\end{asmp}
\begin{asmp}
    \label{asmp:problem-statement:assumption-2-perturbation-signal}
    The dither signal $\dithersignal:\realnumbers\to\realnumbers^n$ is a continuous function, periodic or almost periodic, and
    \begin{equation}
        \label{eq:asmp2:assumption2-equation}
        \sup_{t\in\realnumbers} \Vert \dithersignal(t) \Vert_2 = 1
    \end{equation}
\end{asmp}
Under these assumptions, one can formulate the derivative estimates $\estimate{\derivativevalue}$. Although $\estimate{\derivativevalue}$ may result from differential equations \cite[Eq~14]{ref:nesic-2010-2}, the simplest formulation of $\estimate{\derivativevalue}$ as shown in Fig.~\ref{blk:problem-statement:estimation-block-diagram}, is an element-wise definition 
\begin{equation}
    \label{eq:derivative-estimate:time-wise}
    \estimate{\derivativevalue}_i\left(\omega t, \estimate{\parameter}, a\right) = \derivativefilter_i(\ditherfrequency t, \ditheramplitude) \costfunction\left(\estimate{\parameter} + \ditheramplitude \dithersignal(\ditherfrequency t)\right)
\end{equation}
where the signal $\derivativefilter_i$ is used to "demodulate" the derivative information from the cost function output variation. As expected, $\derivativefilter_i$ depends on the system designer's choice of $\dithersignal$. In \cite{ref:ghaffari-2012}, the dither signal $\dithersignal$ was a sinusoidal signal with elements $\dithersignal_i(t) = d_i \sin(r_i t)$ where $d_i$ and $r_i$ controlled the relative dither amplitudes and frequencies, respectively. This sinusoidal dither signal has corresponding $\derivativefilter_i$ for each $\multiindex_i$ as: $h_i(t, \ditheramplitude) = \frac{2}{\ditheramplitude d_k} \sin(r_k t)$ if $\multiindex_i$ has $k$-th element 1 and the rest zero; $h_i(t, \ditheramplitude) = \frac{16}{\ditheramplitude^2 d_k^2} \left(\sin^2(r_k t) - \frac{1}{2}\right)$ if $\multiindex_i$ has $k$-th element 2 and the rest zero; and $h_i(t, \ditheramplitude) = \frac{4}{\ditheramplitude^2 d_j d_k} \sin(r_j t) \sin(r_k t)$ if $\multiindex_i$ has $j$-th and $k$-th element 1 and the rest zero. One can continue to find $\derivativefilter_i$ for higher than second-order derivatives using \cite{ref:nesic-2010-2} but these results are only for sinusoidal $\dithersignal$ and scalar $\costfunction$. Unlike \cite{ref:ghaffari-2012}, we also consider almost periodic $\dithersignal$. Well-known texts on almost periodic functions \cite{ref:fink-1974,ref:shubin-1978} should be consulted for a thorough treatment of the subject, but for illustrative purposes, an example of an almost periodic $\dithersignal$ is the sinusoidal dither signal with $r_1 = 1$ and $r_2 = \sqrt{2}$ since $\sin(t)$ and $\sin(\sqrt{2}t)$ do not share a common period between them.

The output dimension of $\estimate{\derivativevalue}$ depends on the necessary $m$-th order derivatives required for the method in question. Methods using just the gradient ($m=1$) will have $\estimate{\derivativevalue}(\ditherfrequency t, \estimate{\parameter}, a)\in\realnumbers^{n}$ whereas Newton-based methods ($m=2$) need both the gradient and the Hessian of $\costfunction$ and thus $\estimate{\derivativevalue}(\ditherfrequency t, \estimate{\parameter}, a)\in\realnumbers^{\frac{1}{2}n^2 + \frac{3}{2}n}$. To generalize, we define $\multiindexsize(k)$ as
\begin{equation}
    \multiindexsize(k) = \big\vert \left\lbrace \multiindex\ \vert\ \left\vert\multiindex\right\vert = k \right\rbrace \big\vert = \binom{n + k - 1}{k}  
\end{equation}
to keep count of the total number of unique $k$-th order partial derivatives. Now we can state that $\estimate{\derivativevalue}(\ditherfrequency t, \estimate{\parameter}, a)\in\realnumbers^{o}$ where $o = \sum_{i=0\text{ or }1}^{m}\multiindexsize(i)$ with the initial $i$ depending on if the algorithm requires $D^{0}\costfunction = \costfunction$. 

The analysis of ESCs as mentioned previously, does not happen in the original system, but rather in some nominal system, which is the average system for this work. The average system is defined as
\begin{align}
    \label{eq:parameter-update-ode:average}
    \derivative{t}\avgest{\state} &= f_0\left(\avgest{\state} \right) + f_1\left(\avgest{\state}\right) \avgest{\derivativevalue}\left(\avgest{\parameter}, \ditheramplitude\right)
\end{align}
where $\avgest{\state}$  and $\avgest{\parameter}$ are the average state and parameter estimates with the average derivative estimate $\avgest{\derivativevalue}$ defined as
\begin{equation}
    \label{eq:derivative-estimate:average}
    \avgest{\derivativevalue}_i\left(\avgest{\parameter}, a\right)  = \lim_{T\to\infty}\frac{1}{T}\int_{0}^{T} \derivativefilter_i(\tau, a) \costfunction\left(\avgest{\parameter} + \ditheramplitude \dithersignal(\tau)\right) d\tau
\end{equation}
One would hope that $\lim_{\ditheramplitude\to 0} \avgest{\derivativevalue}_i(\avgest{\parameter}, \ditheramplitude)=\differentialoperator{\multiindex_i}\costfunction(\avgest{\parameter})$ for all $i$ so that the average system resembles some model-based system
\begin{align}
    \label{eq:parameter-update-ode:model-based}
    \derivative{t}\altstate &= f_0\left(\altstate \right) + f_1\left(\altstate\right) \derivativevalue(\altparameter)
\end{align}
for sufficiently small $\ditheramplitude$. Here, the adjective ``model-based'' indicates that the system has direct access to the derivative information of $\costfunction$ and we use $\altstate$ and $\altparameter$ to represent the state vector and parameter estimates in this model-based system. One would like to study the model-based system to glean stability insights into the original model-free system. However, we first need to find an $\derivativefilter$ if even one exists so that one can relate the model-free systems with their respective model-based systems.

\section{Main Results}
    \label{sec:main-results}
    In order to show the existence of an $\demodulationsignal$, we first consider the problem of finding various derivatives $\derivativevalue$ for a perturbed polynomial $P_m \left(\ditheramplitude \dithersignal(t), \avgest{\parameter}\right)$ of order $m$ which is centered at $\avgest{\parameter}$. This polynomial can be expressed in multi-index notation as 
\begin{align}
    \label{eq:proof:polynomial-multi-index}
    P_m(\ditheramplitude \dithersignal(t),\avgest{\parameter}) &= \sum_{|\multiindex|\leq m}\frac{1}{\multiindex!}a^{|\multiindex|}\dithersignal(t)^{\multiindex} \derivativevalue_i\left(\avgest{\parameter}\right) \\
    &= \extendeddithersignal(t)^T A \derivativevalue\left(\avgest{\parameter}\right)
\end{align}
where $\derivativevalue_i\left(\avgest{\parameter}\right)=D^{\multiindex_i} \costfunction(\avgest{\parameter})$ are the derivative values of interest, the extended dither signal vector includes the signals $\extendeddithersignal_i(t)~=~\dithersignal(t)^{\multiindex_i}$ associated with the $\derivativevalue_i$, and $A$ is a diagonal matrix with $\ditheramplitude^{\vert \multiindex_i \vert}/\multiindex_i !$ as the $i$-th element on the main diagonal and zero everywhere else. Note that $\extendeddithersignal$ is a continuous (almost) periodic function as a result of $\dithersignal$ being a continuous (almost) periodic function in Assumption \ref{asmp:problem-statement:assumption-2-perturbation-signal}. Also, since $\sup_{t\in\realnumbers} \Vert \dithersignal(t) \Vert_2= 1$ by Assumption \ref{asmp:problem-statement:assumption-2-perturbation-signal}, the dither signal components being upper-bounded by $\left\vert \dithersignal(t) \right\vert \leq 1$ implies that $\left\vert \extendeddithersignal_i(t) \right\vert \leq 1$ and $\Vert \extendeddithersignal(t) \Vert_2 \leq \sqrt{o}$ for all $t\in\realnumbers$.

The average derivative estimate for this polynomial is
\begin{align}
    \label{eq:proof:average-estimate-polynomial}
    \avgest{\derivativevalue}\left(\avgest{\parameter}, \ditheramplitude\right) &= \lim_{T\to\infty}\frac{1}{T}\int_0^T \demodulationsignal(\tau, \ditheramplitude)P_m(\ditheramplitude\dithersignal(\tau),\avgest{\parameter})d\tau \\
    \label{eq:proof:average-estimate-polynomial:simplified}
    &{}={} \left(\lim_{T\to\infty}\frac{1}{T}\int_0^T \demodulationsignal(\tau, \ditheramplitude) \extendeddithersignal(\tau)^T d\tau\right) A \derivativevalue\left(\avgest{\parameter}\right)
\end{align}
and the desired outcome is that $\demodulationsignal$ results in $\avgest{\derivativevalue}\left(\avgest{\parameter},\ditheramplitude\right) = \derivativevalue\left(\avgest{\parameter}\right)$ for any nonzero $\ditheramplitude$. Thus, $\demodulationsignal$ must satisfy the condition
\begin{equation}
    \label{eq:proof:demodulation-signal:perturbation-requirement}
    \lim_{T\to\infty}\frac{1}{T} \int_{0}^T \demodulationsignal(\tau, \ditheramplitude) \extendeddithersignal(\tau)^T d\tau = A^{-1}
\end{equation}
so that $\demodulationsignal$ achieves the desired result $\avgest{\derivativevalue}\left(\avgest{\parameter}, \ditheramplitude\right) = \derivativevalue\left(\avgest{\parameter}\right)$. Here we can tell that a necessary and sufficient condition for the existence of $\demodulationsignal$ is that $\extendeddithersignal$ is a vector of linearly independent signals, which is to say that
\begin{equation}
    \label{eq:proof:beta-condition}
    \nexists \beta \neq 0\ \suchthat\ \extendeddithersignal(t)^T \beta = 0\ \forall t\in\realnumbers
\end{equation}
The condition is necessary since, if there existed a $\beta \neq 0$, then right multiplying both sides of \eqref{eq:proof:demodulation-signal:perturbation-requirement} by $\beta$ would result in the contradiction of the left side being a zero vector but the right side a non-zero vector. The condition is also sufficient since
\begin{equation}
    \label{eq:proof:demodulation-equation}
    \demodulationsignal(t, \ditheramplitude) = A^{-1}\covariance^{-1}\extendeddithersignal(t)
\end{equation}
with $\covariance$ being the covariance matrix of $\extendeddithersignal$ as defined by
\begin{equation}
\covariance = \lim_{T\to\infty}\frac{1}{T}\int_0^T \extendeddithersignal(\tau) \extendeddithersignal(\tau)^T d\tau
\end{equation}
would satisfy \eqref{eq:proof:demodulation-signal:perturbation-requirement}. The demodulation signal $\demodulationsignal$ is well defined since $\covariance$ is a positive definite matrix, and therefore invertible, since
\begin{align}
    \beta^T \covariance \beta & = \lim_{T\to\infty} \frac{1}{T}\int_{0}^{T} \left(\beta^T \extendeddithersignal(\tau)\right)^2 d\tau \\
    &\leq \lambda_{\max}\left(\covariance\right) \Vert \beta \Vert_2^2 \leq \trace\left(\covariance\right) \Vert \beta \Vert_2^2 \leq o \left\Vert \beta \right\Vert_2^2
\end{align}
is finite and strictly positive \cite[Th~3.8]{ref:fink-1974} for any nonzero $\beta$. Note that to meet the condition in \eqref{eq:proof:beta-condition}, $\extendeddithersignal$ only needs to be comprised of linearly independent signals, not necessarily orthogonal signals as stated in \cite{ref:nesic-2010-2}. This implies, particularly for the $m=1$ case, that it is not necessary for $\dithersignal$ to be comprised of orthogonal signals when estimating the derivative as required by \cite{ref:mills-2018}.

The extended dither function is linearly independent if and only if $\covariance$ is invertible. In practice, it is easier to calculate $\covariance$ and check its determinate to determine if $\covariance$ is invertible than to check the linear independence of $\extendeddithersignal$ by calculating an orthogonal basis using the Gram-Schmidt reduction process. This is because: 1) we already need $\covariance^{-1}$ for \eqref{eq:proof:demodulation-equation} and 2) the computational overhead to determine and simplify the signals in the orthogonal basis can be quite complex.

Additionally, we do not claim that \eqref{eq:proof:demodulation-equation} results in the only possible solution to \eqref{eq:proof:demodulation-signal:perturbation-requirement} since it may be possible to find another $\demodulationsignal'$ with
\begin{equation}
    \label{eq:proof:crossvariance-demodulation}
    \demodulationsignal'(t,\ditheramplitude) = A^{-1} R^{-1} r(t)
\end{equation}
where $r(t)\in\realnumbers^o$ is vector of continuous, bounded, linear independent, (almost) periodic signals and $\crossvariance$ is the cross-variance between the signals
\begin{equation}
    \label{eq:proof:crossvariance-covariance-matrix}
    \crossvariance = \lim_{\period{T}\to \infty}\frac{1}{\period{T}}\int_{0}^\period{T} r(\tau) \extendeddithersignal(\tau)^T d\tau
\end{equation}
This requires finding $r$ and ensuring that the associated $\crossvariance$ is invertible, a task which may not be trivial. Hence, determining $\demodulationsignal$ is easier with the covariance version \eqref{eq:proof:demodulation-equation} rather than the cross-variance version \eqref{eq:proof:crossvariance-demodulation}.

Having determined the necessary and sufficient conditions for the existence of $\demodulationsignal$, we present the main theorem of this work which states that, for a non-polynomial $\costfunction$, the limit of $\avgest{\derivativevalue}(\altparameter, \ditheramplitude) \to \derivativevalue(\altparameter)$ is valid as $\ditheramplitude$ approaches zero.
\begin{thm}
    \label{thm:existence-of-consistent-derivative-estimation-signals}
    Let $\costfunction: \realnumbers^n \to \realnumbers$ satisfies Assumption \ref{asmp:problem-statement:assumption-1-cost-function} for an $m$-th order ESC and $\dithersignal$ satisfies Assumption \ref{asmp:problem-statement:assumption-2-perturbation-signal}. Possible demodulation signal vectors $\demodulationsignal$ exist if and only if the associated $\extendeddithersignal$ is a vector of linearly independent signals. The use of $h$ in \eqref{eq:derivative-estimate:time-wise} will result in the limit $\avgest{\derivativevalue}(\avgest{\parameter}, \ditheramplitude)\to\derivativevalue(\avgest{\parameter})$ as $\ditheramplitude\to 0$ for all $\avgest{\parameter}\in\realnumbers^n$.
\end{thm}
The condition for the existence of $\demodulationsignal$ has already been shown and so we prove the convergence portion of Theorem~\ref{thm:existence-of-consistent-derivative-estimation-signals}.

\begin{proof}
    In order to prove convergence of $\avgest{\derivativevalue}$, we need to use an $\epsilon-\delta$ definition of convergence
    \begin{multline}
        \label{eq:proof:convergence:epsilon-delta-definition}
        \forall\avgest{\parameter}\in\realnumbers^{n},\forall\epsilon > 0,\ \exists \delta = \delta\left(\avgest{\parameter},\epsilon\right) > 0\ \\ \suchthat\ \forall a\ \in (0,\delta),\ \left\Vert\avgest{\derivativevalue}(\avgest{\parameter}, \ditheramplitude)-\derivativevalue(\avgest{\parameter})\right\Vert_2 < \epsilon
    \end{multline}
    where the estimate error $\left\Vert \avgest{\derivativevalue}(\avgest{\parameter}, \ditheramplitude)-\derivativevalue(\avgest{\parameter})\right\Vert_2$ is the Euclidean norm
    \begin{equation}
        \label{eq:proof:functional-norm}
        \left\Vert \avgest{\derivativevalue}(\avgest{\parameter}, \ditheramplitude)-\derivativevalue(\avgest{\parameter}) \right\Vert_2 = \left(\sum_{i=1}^o\left\vert \avgest{\derivativevalue}_i(\avgest{\parameter}, \ditheramplitude)-\derivativevalue_i(\avgest{\parameter}) \right\vert^2\right)^{1/2}
    \end{equation}
    Furthermore, we can focus on the individual average derivative estimates by finding every $\delta_i = \delta_i\left(\avgest{\parameter}, \varepsilon/\sqrt{o}\right) > 0$ such that
    \begin{equation}
        \label{eq:proof:convergence:component-tolerance}
        \left\vert \avgest{\derivativevalue}_i\left(\avgest{\parameter}, a\right) - \derivativevalue_i\left(\avgest{\parameter}\right) \right\vert < \frac{\epsilon}{\sqrt{o}}
    \end{equation}
    for all $\ditheramplitude\in(0, \delta_i)$ as $\delta = \min_{i=1,\ldots,o} \delta_i$ would be sufficient for \eqref{eq:proof:convergence:epsilon-delta-definition}.

    For the individual derivative estimates, we calculate the difference between the average estimate and the true derivative value with
    \begin{multline}
    \label{eq:proof:convergence:average-estimate-error}
        \avgest{\derivativevalue}_i\left(\avgest{\parameter}, \ditheramplitude\right) - \derivativevalue_i\left(\avgest{\parameter}\right) = \lim_{\period{T}\to\infty}\frac{1}{T}\int_0^{T}\demodulationsignal_i(\tau, a) \\ \cdot\left[ \costfunction\left(\avgest{\parameter} +\ditheramplitude \dithersignal(\tau)\right) - 
        P_{m}\left(\ditheramplitude\dithersignal(\tau),\avgest{\parameter}\right)\right]d\tau
    \end{multline}    
    by utilizing the fact that $\avgest{\derivativevalue}_i\left(\avgest{\parameter}, \ditheramplitude\right) = \derivativevalue_i\left(\avgest{\parameter}\right)$ for the polynomial $P_m$ by our construction of $\demodulationsignal$. Utilizing Taylor's Theorem \cite[Th~1]{ref:ash-1990}, we can further simplify \eqref{eq:proof:convergence:average-estimate-error} to
    \begin{multline}
    \label{eq:proof:convergence:average-estimate-error-2}
        \avgest{\derivativevalue}_i\left(\avgest{\parameter}, \ditheramplitude\right) - \derivativevalue_i\left(\avgest{\parameter}\right) = \lim_{\period{T}\to\infty}\frac{1}{T}\int_0^{T}\demodulationsignal_i(\tau, a) \\ \cdot\left[ \sum_{\vert \multiindex_j \vert = m} \frac{a^{m}}{\multiindex_j!}\left(D^{\multiindex_j}\costfunction\left(\avgest{\parameter} + \eta_j(\tau)\right) - D^{\multiindex_j}\costfunction\left(\avgest{\parameter}\right)\right) \extendeddithersignal_j(\tau) \right]d\tau
    \end{multline}
    where $\eta_j(\tau)\in\ball{\ditheramplitude}$ for all $\tau\in (0, T)$. We now bound the magnitude of the error by
    \begin{multline}
    \label{eq:proof:convergence:average-estimate-error-magnitude}
        \left\vert \avgest{\derivativevalue}_i\left(\avgest{\parameter}, \ditheramplitude\right) - \derivativevalue_i\left(\avgest{\parameter}\right) \right\vert \leq \sum_{\vert \multiindex_j \vert = m}\sup_{\tau\in\realnumbers} \left\vert \demodulationsignal_i(\tau, a)\right\vert \\ \cdot\left[ \frac{a^{m}}{\multiindex_j!}\left\vert D^{\multiindex_j}\costfunction\left(\avgest{\parameter} + \eta_j(\tau)\right) - D^{\multiindex_j}\costfunction\left(\avgest{\parameter}\right)\right\vert \left\vert \extendeddithersignal_j(\tau) \right\vert \right]
    \end{multline}
    The demodulation signal component, when formulated with \eqref{eq:proof:crossvariance-demodulation}, can be bounded by
    \begin{align}
        \left\vert \demodulationsignal_i(\tau) \right\vert^2 &\leq  \left\Vert \demodulationsignal(\tau) \right\Vert_2^2 = r(\tau)^T \left(\crossvariance^{-1}\right)^T A^{-2} \crossvariance^{-1} r(\tau) \\
        &\leq \sigma_{o}\left(A\right)^{-2} r(\tau)^T \left(\crossvariance^{-1}\right)^T \crossvariance^{-1} r(\tau) \\
        &\leq \sigma_{o}\left(A\right)^{-2} \sigma_{o}\left(\crossvariance\right)^{-2} \left\Vert r(\tau) \right\Vert_2^2
    \end{align}
    where $\sigma_o(\cdot)$ gives the $o$-th, i.e.,\ the smallest, singular value of a matrix. For $A$, the smallest singular value is known to be $\ditheramplitude^{m}/m!$ when $\ditheramplitude < 1$ while $\sigma_o(\crossvariance)$ is some constant dependent on the selection of $r$. Since $r$ is a boundable (almost) periodic function and $\exists\ r_{\max} > 0$ such that $\sup_{t\in\realnumbers}\Vert r (t) \Vert \leq r_{\max} < \infty$, then \eqref{eq:proof:convergence:average-estimate-error-magnitude} has the upper-bound
    \begin{multline}
    \label{eq:proof:convergence:average-estimate-error-magnitude-2}
        \left\vert \avgest{\derivativevalue}_i\left(\avgest{\parameter}, \ditheramplitude\right) - \derivativevalue_i\left(\avgest{\parameter}\right) \right\vert \leq  \frac{m! r_{\max}}{\sigma_o(\crossvariance)} \\ \cdot \left(\sum_{\vert \multiindex_j \vert = m}\sup_{\tau\in\realnumbers} \left\vert D^{\multiindex_j}\costfunction\left(\avgest{\parameter} + \eta_j(\tau)\right) - D^{\multiindex_j}\costfunction\left(\avgest{\parameter}\right)\right\vert \right)
    \end{multline}
    This gives bound with $h$ formulated by \eqref{eq:proof:crossvariance-demodulation}. If $\demodulationsignal$ had been formulated by \eqref{eq:proof:demodulation-equation} instead, then one should replace $r_{\max}$ and $\crossvariance$ with $\sqrt{o}$ and $\covariance$, respectively.
    
    Since $\costfunction\in\continuous^m$ by Assumption \ref{asmp:problem-statement:assumption-1-cost-function}, all derivatives $D^{\multiindex_j}\costfunction$ are continuous functions. Therefore we can make the components of the summation in the right hand side of \eqref{eq:proof:convergence:average-estimate-error-magnitude-2} to be as small as possible. Let $\delta_j$ be the $\delta$ function in the $\varepsilon-\delta$ continuity definition for the $j$-th derivative $D^{\multiindex_j}\costfunction$ where $\vert \multiindex_j \vert = m$. Then
    \begin{equation}
        \delta_i\left(\avgest{\parameter}, \frac{\varepsilon}{\sqrt{o}}\right) = \min_{j} \delta_j\left(\avgest{\parameter}, \frac{\varepsilon \sigma_o(\crossvariance)}{\multiindexsize(m)\cdot r_{\max}\cdot m! \cdot \sqrt{o}} \right)
    \end{equation}
    ensures that the right hand side of \eqref{eq:proof:convergence:average-estimate-error-magnitude-2} meets the condition in \eqref{eq:proof:convergence:component-tolerance}. Thus 
    the $\delta$ function for the definition in \eqref{eq:proof:convergence:epsilon-delta-definition} is
    \begin{equation}
        \delta\left(\avgest{\parameter}, \varepsilon\right) = \min_{j=1,\ldots, c(m)}\delta_j \left(\avgest{\parameter}, \frac{\varepsilon \sigma_o(\crossvariance)}{\multiindexsize(m)\cdot r_{\max}\cdot m! \cdot \sqrt{o}} \right)
    \end{equation}
    which completes the proof.
\end{proof}

\begin{rem}
    \label{rmk:convergence-derivatives}
    It is important to show the convergence in the derivative estimates because there are cases where the residual terms cause instabilities in the average system as shown in \cite{ref:mcnamee-2025:scalar-newton-automatica}.
\end{rem}

\begin{rem}
    While the convergence proof of Theorem~\ref{thm:existence-of-consistent-derivative-estimation-signals} seems more complicated than necessary with the use of the $\varepsilon-\delta$ proof, this is a necessary level of complexity given Assumption \ref{asmp:problem-statement:assumption-1-cost-function}. 
    Let $\kappa\in\classK$ be the slowest local convergence rate for the derivatives $D^{\multiindex_i}\costfunction$ so that
    \begin{equation}
        \left\vert D^{\multiindex_i}\costfunction(\parameter) - D^{\multiindex_i}\costfunction\left(\avgest{\parameter}\right) \right\vert \in \mathcal{O}\left(\kappa\left(\left\Vert \parameter - \avgest{\parameter}\right\Vert\right)\right)
    \end{equation}
    for all $\vert\multiindex_i\vert = m$ and $\parameter$ in a neighborhood of $\avgest{\parameter}$. By an examination of \eqref{eq:proof:convergence:average-estimate-error-magnitude-2}, it is clear that the derivative estimates will share the same local convergence rate of the worst $m$-th order derivative local convergence rate. If $\costfunction\in\continuous^{m+1}$, like in \cite{ref:rusiti-2016} and \cite{ref:mills-2018}, then the local convergence rate is linear with $\kappa(r) = r$ but the local convergence is not necessarily linear. Consider the scalar cost function
    \begin{equation}
        \costfunction(\parameter) = \frac{4}{15}\left\vert \parameter \right\vert^{5/2}
    \end{equation}
    whose second-order derivative (i.e. the Hessian) is $\frac{d^2 \costfunction}{d\parameter^2}~=~\sqrt{\left\vert \parameter \right\vert}$. The second-order derivative satisfies the constraint
    \begin{equation}
        \left\vert \frac{d^2 \costfunction}{d\parameter^2}(\parameter) - \frac{d^2 \costfunction}{d\parameter^2}(0) \right\vert \leq \sqrt{\left\vert \parameter \right\vert}
    \end{equation}
    around $\parameter = 0$ which means that the second-order derivative estimate should have a convergence rate of $\kappa(r) = \sqrt{r}$, which is slower than linear. We show in Fig.~\ref{fig:1d-example-convergence} the convergence of the average Hessian estimate \cite{ref:ghaffari-2012}
    \begin{equation}
        \avgest{\hessian}(0, \ditheramplitude) = \frac{1}{2\pi}\int_0^{2\pi}\frac{-8}{\ditheramplitude^2}\cos(2\tau)\costfunction(\ditheramplitude\sin(\tau)) d\tau
    \end{equation}
    as a function of $\ditheramplitude$ in a log scale along with the functions $\ditheramplitude$ and $\sqrt{\ditheramplitude}$. The convergence rate can be gleaned by the slope in the log scale and $\avgest{\hessian}(0, \ditheramplitude)$ matches the shallower slope of $\sqrt{\ditheramplitude}$ rather than the steeper $\ditheramplitude$. Thus, $\avgest{\hessian}$ locally converges with respect to $\ditheramplitude$ at a rate \emph{slower} than linear.
    \begin{figure}
        \centering
        \includegraphics[width=1\linewidth]{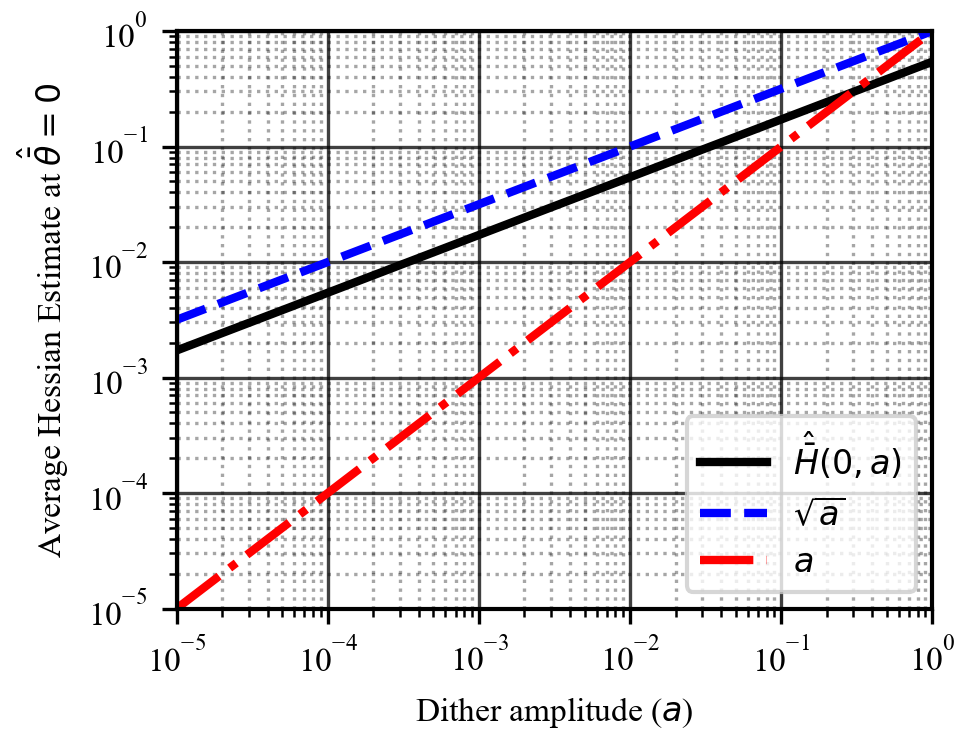}
        \caption{
        1D example of the convergence explained by the proof in more detail.
        }
        \label{fig:1d-example-convergence}
    \end{figure}
\end{rem}

While Theorem~\ref{thm:existence-of-consistent-derivative-estimation-signals} gives the conditions for the existence of $\demodulationsignal$ whose signal components may be non-zero mean, it is possible to estimate the other derivatives without necessarily estimating $\costfunction(\avgest{\parameter})$.
\begin{prop}
    Let $\costfunction: \realnumbers^n \to \realnumbers$ satisfy Assumption \ref{asmp:problem-statement:assumption-1-cost-function} for an $m$-th order ESC and $\dithersignal$ satisfy Assumption \ref{asmp:problem-statement:assumption-2-perturbation-signal}. If one wants to use strictly zero mean signals to estimate derivatives of order $1$ through $m$ without estimating $\costfunction(\estimate{\parameter})$ ($m=0$), then possible demodulation signal vectors $\demodulationsignal$ exist if and only if the associated $ \error{\extendeddithersignal}$ is a vector of linearly independent signals where $\error{\extendeddithersignal}(t)~=~\extendeddithersignal(t)~-~\average{\extendeddithersignal}$ with
        $\average{\extendeddithersignal} = \lim_{T\to\infty}\frac{1}{T}\int_{0}^T \extendeddithersignal(\tau) d\tau$
    is linearly independent. The use of $h$ in \eqref{eq:derivative-estimate:time-wise} will result in the limit $\avgest{\derivativevalue}(\avgest{\parameter}, \ditheramplitude)\to\derivativevalue(\avgest{\parameter})$ as $\ditheramplitude\to 0$ for all $\avgest{\parameter}\in\realnumbers^n$.
\end{prop}
\begin{proof}
    Consider again the average derivative estimate
    \begin{align}
        \avgest{\derivativevalue}\left(\avgest{\parameter}, \ditheramplitude\right) &= \lim_{T\to\infty}\frac{1}{T}\int_0^T \demodulationsignal(\tau,\ditheramplitude) P_m(\ditheramplitude\dithersignal(\tau), \avgest{\parameter}) d\tau
    \end{align}
    \begin{multline}
        \phantom{\avgest{\derivativevalue}\left(\avgest{\parameter}, \ditheramplitude\right)} = \left(\lim_{T\to\infty}\frac{1}{T}\int_0^T \demodulationsignal(\tau,\ditheramplitude)\costfunction\left(\avgest{\parameter}\right)d\tau\right) \\
        + \left(\lim_{T\to\infty}\frac{1}{T}\int_0^T \demodulationsignal(\tau,\ditheramplitude)\extendeddithersignal(\tau)^T d\tau\right)A\derivativevalue\left(\avgest{\parameter}\right)
    \end{multline}
    where $\estimate{\derivativevalue}$ are estimates of the derivatives $\derivativevalue$ which now range from order 1 to $m$. By $\demodulationsignal$ being a zero mean signal, this average estimate is simplified to
    \begin{equation}
        \label{eq:zero-mean-proposition:average-derivative-estimate}
        \avgest{\derivativevalue}\left(\avgest{\parameter}, \ditheramplitude\right) = \left(\lim_{T\to\infty}\frac{1}{T}\int_0^T \demodulationsignal(\tau,\ditheramplitude)\error{\extendeddithersignal}(\tau)^T d\tau\right)A\derivativevalue\left(\avgest{\parameter}\right)
    \end{equation}
    since $\extendeddithersignal(\tau)~=~\error{\extendeddithersignal}(\tau)~-~\average{\extendeddithersignal}$. Since \eqref{eq:zero-mean-proposition:average-derivative-estimate} is now the same as \eqref{eq:proof:average-estimate-polynomial:simplified} with $\error{\extendeddithersignal}$ instead of $\extendeddithersignal$, it follows that the necessary and sufficient condition for the existence of $\demodulationsignal$ is the linear independence of $\error{\extendeddithersignal}$. The proof of the limit $\avgest{\derivativevalue}(\avgest{\parameter}, \ditheramplitude)\to\derivativevalue(\avgest{\parameter})$ as $\ditheramplitude\to 0$ is then identical to that in Theorem~\ref{thm:existence-of-consistent-derivative-estimation-signals}.
\end{proof}
\color{black}
\section{Nonholonomic Ground Vehicle Controller Design}
    \label{sec:nonholonomic-ground-vehicle-ex}
    To illustrate the usefulness of our results, we consider the problem of seeking a static source using a nonholonomic unicycle model. The vehicle dynamics is represented as
\begin{align}
    \dot{x}_c &= \commanded{\forwardvelocity} \cos\headingangle \\
    \dot{y}_c &= \commanded{\forwardvelocity} \sin\headingangle \\
    \dot{\headingangle} &= \commanded{\angularvelocity}
\end{align}
where $x_c$ and $y_c$ are the position of the vehicle's center in the global frame, $\headingangle$ is the heading angle, and $\commanded{\forwardvelocity}$ and $\commanded{\angularvelocity}$ are the commanded forward and angular velocities. The vehicle is also equipped with a rotating microphone sensor arm mounted on top of the robot as shown in Fig.\ \ref{fig:problem-statement:unicylce-model}. The concept behind this vehicle is to use local derivative estimates in the relative $(x_r, y_r)$ frame to integrate into a controller to move the vehicle towards the extremum.
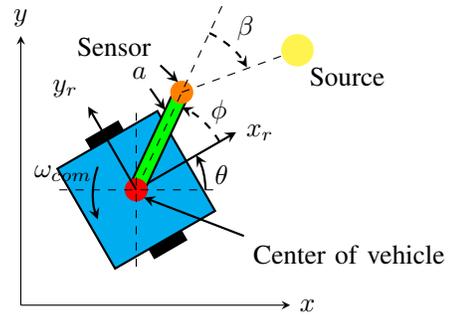
\begin{figure}[t!]
    \centering
    \resizebox{0.35\textwidth}{!}{
        \begin{tikzpicture}[scale=1, >=stealth]

    \draw[->] (1.5, 1.5) -- (5, 1.5) node[right] {$x$};
    \draw[->] (1.5, 1.5) -- (1.5, 5) node[above] {$y$};
    
    \begin{scope}[rotate around={30:(3,3)}]
        \fill[black] (2.75, 2.1) rectangle (3.25,3.9);
        \fill[cyan] (2.25, 2.25) rectangle (3.75, 3.75);
        \draw[thick] (2.25, 2.25) rectangle (3.75, 3.75);
        \draw[->, thick] (3, 3) -- (4.5, 3) node[anchor=west] {$x_r$};
        \draw[dashed,->, thick] (4.25, 3) arc[start angle=0, end angle=35, radius=1.1];
        \node[above] at (4.3,3.1) {$\phi$};
    \end{scope}

    \begin{scope}[rotate around={65:(3,3)}]
        \fill[green] (3,2.9) rectangle (4.3, 3.1);
        \draw[thick] (3,2.9) rectangle (4.3, 3.1);
        \node[right] at (4.3, 3.8) (armlength) {$\ditheramplitude$};
        \draw[->,thick] (armlength) -- (4.1,3.1);  
        \fill[orange] (4.4, 3) circle (0.15);
        \begin{scope}[rotate around={-45:(4.4, 3)}]
            \draw[dashed] (4.4, 3) --++ (1.5, 0);
            \fill[color=yellow!80!white] (6, 3) circle (6pt) node[anchor=north west, color=black] at (6,2.9) {Source};
        \end{scope}
        \draw[dashed, thick, ->] (5.25, 3) arc[start angle=0, end angle=-35, radius=1.1];
        \node at  (5.5,2.6) {$\sensorangletosource$};
        \draw[dashed] (3, 3) --++ (2.75, 0);
        \node[anchor=south east] at (4.6,3.4) (sensor) {Sensor};
        \draw[->,thick] (sensor.south east) -- (4.5,3.1);
    \end{scope}

    \fill[red] (3,3) circle (0.15);
    \node[anchor=north west] at (4.4, 2.4) (vehicle_center) {Center of vehicle};
    \draw[->,thick] (vehicle_center.north west) -- (3.1,2.9);
    \draw[->, thick] (3.9, 3) arc[start angle=0, end angle=30, radius=.9];
    \node[right] at (3.9,3.2) {$\headingangle$};

    \begin{scope}[rotate around={120:(3,3)}]    
        \draw[thick, ->] (3,3) -- (4.25,3) node[anchor=south east] {$y_r$};
    \end{scope}

    \draw[->, thick] (2.5, 3.3) arc[start angle=160, end angle=200, radius=1];
    \node[left] at (2.55, 3.22) {$\commanded{\angularvelocity}$};

    \draw[dashed] (2, 3) -- (4, 3);
    \draw[dashed] (3,2) -- (3, 4);

\end{tikzpicture}
    }
    \caption{Unicycle model with a rotating sensor. Note that $\sensorangletosource$ as shown is negative.}
    \label{fig:problem-statement:unicylce-model}
\end{figure}

\subsection{Design of Estimation Filters}

Previous work would have attempted to use sinusoidal perturbations of the sensor, achievable when $\microphoneangle$ is unconstrained, to locally estimate the derivatives in the relative frame. However, consider implementation constraints, such as requiring $\phi(t)\in [-\pi/2, \pi/2]$ for all $t$. This constraint prevents using sinusoidal perturbations, but we want to find perturbation signals $\dithersignal$ which have $\demodulationsignal$'s for estimating the derivatives in the relative frame. If we can find these $\dithersignal$, then we can avoid using the vehicle itself to perturb the map for perturbation-based derivative estimation as was the case with previous approaches \cite{ref:zhang-2007, ref:cochran-2009-2}. This is beneficial as delegating the sensor exploration to a servo enables energy and time savings during convergence to the source since the vehicle does not have to locally explore the space while also simultaneously moving towards the source.

\begin{figure}[t!]
    \centering
    \subfloat[]{
        \includegraphics[width=3.25in]{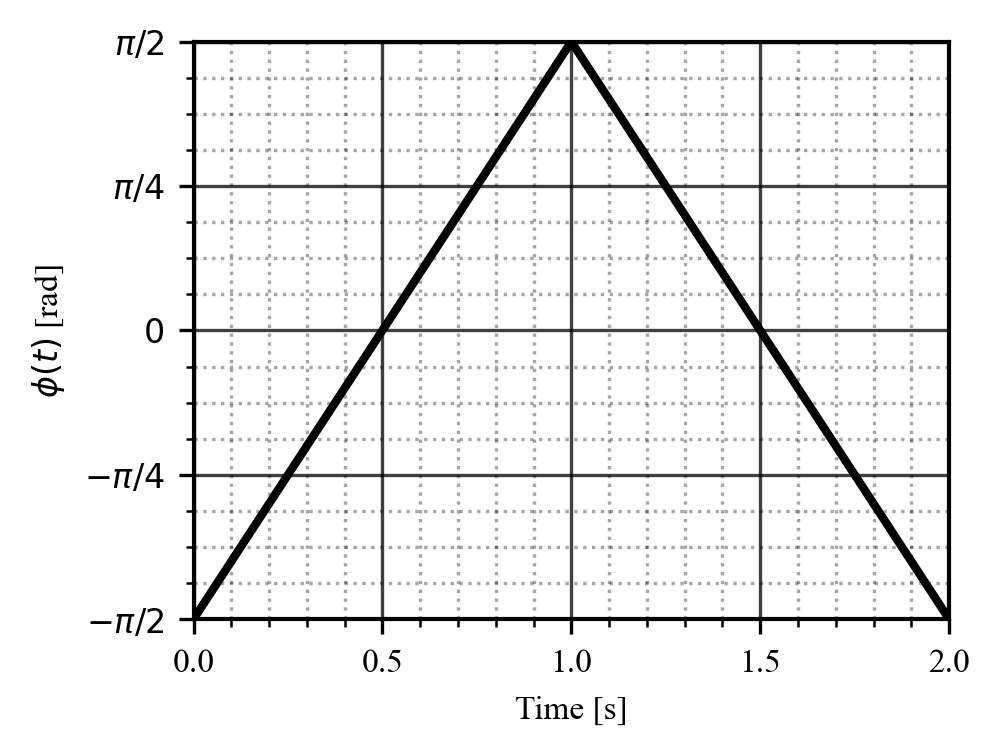}%
        \label{subfig:phi-function-v-time}
    }
    \hfil
    \subfloat[]{
        \includegraphics[width=3.25in]{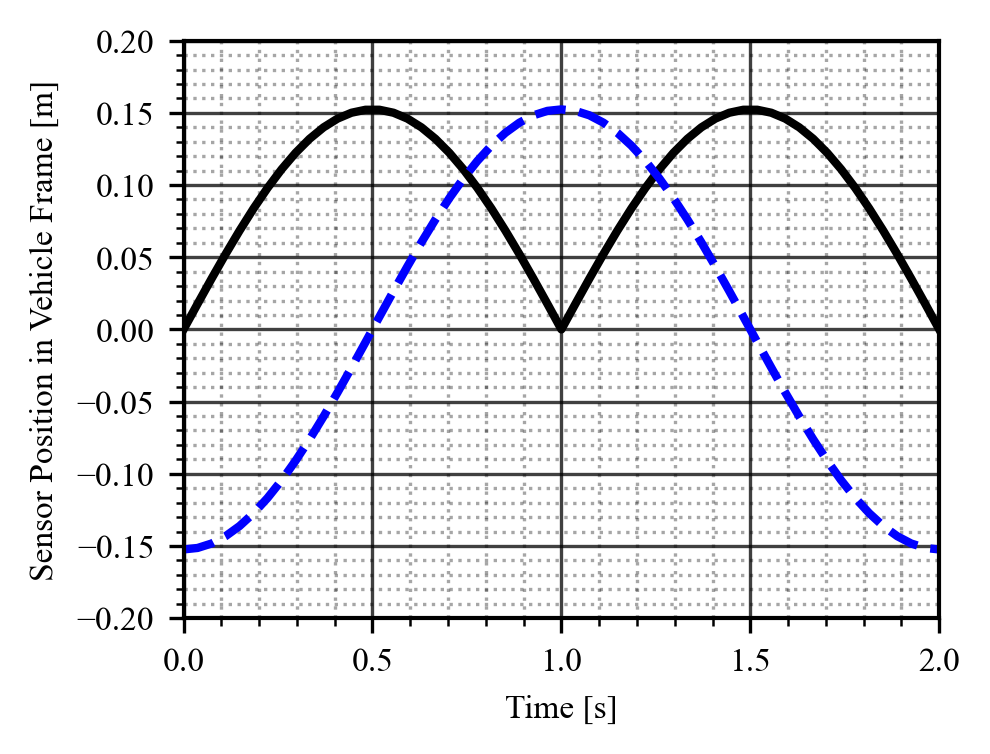}%
        \label{subfig:sensor-x-y-position-v-time}
    }
    \caption{
        (a) The $\microphoneangle(t)$ function over time and (b) The sensor position in the $x$ $y$ relative frame over time. 
    \\
    \textbf{KEY for b)}: 
    $p_{x_r}=$ \protect\raisebox{1.5pt}{\protect\tikz \protect\draw[black, line width=3pt] (0,0) -- (0.5, 0);} ,
    $p_{y_r}=$ \protect\raisebox{1.5pt}{\protect\tikz \protect\draw[blue,dashed, line width =3pt] (0,0) -- (0.5,0);}.
    }
    \label{fig:example:phi-function-px-py}
\end{figure}

Consider the sensor arm angle $\microphoneangle$ shown in Fig.~\ref{fig:example:phi-function-px-py} and defined as
\begin{equation}
\label{eq:example:phi-piecewise-function}
    \microphoneangle(t) = 
    \begin{cases}
        -\pi/2 + \pi\cdot(t\bmod 2), & \text{if}\ t\bmod2 < 1,\\[3mm]
        \ \pi/2 - \pi\cdot ((t-1)\bmod 2),  & \text{otherwise.}\
    \end{cases}
\end{equation}
which is a periodic function with a period of 2. Perturbations of the sensor in the relative frame is then
\begin{equation}
\label{eq:example:dither-amplitude}
    \ditheramplitude \dithersignal(t) = a\begin{bmatrix} \dithersignal_{x_r}(t) \\ \dithersignal_{y_r}(t) \end{bmatrix} = a\begin{bmatrix} \cos(\microphoneangle(t)) \\ \sin(\microphoneangle(t)) \end{bmatrix}
\end{equation}
where $\dithersignal_{x_r}$ and $\dithersignal_{y_r}$ are the perturbations in the relative frame of the vehicle. For illustrative purposes, we examine whether it is possible to estimate the gradient and Hessian of $\costfunction$ in the relative frame. The extended dither signal $\extendeddithersignal(t)$ for the first and second derivatives is
\begin{equation}
    \extendeddithersignal(t) = \begin{bmatrix}
        \cos(\phi(t)) \\
        \sin(\phi(t)) \\
        \cos^2(\phi(t)) \\
        \sin(\phi(t))\cos(\phi(t)) \\
        \sin^2(\phi(t))
    \end{bmatrix}
\end{equation}
The average of this specific signal vector is
\begin{equation}
    \label{eq:example:average-states}
        \average{\extendeddithersignal} = \lim_{T\to\infty}\frac{1}{\period{T}}\int_0^\period{T}\extendeddithersignal(\tau) d\tau= \begin{bmatrix}
        2/\pi\\
        0\\
        1/2\\
        0\\
        1/2
    \end{bmatrix}
\end{equation}
The covariance matrix of $\error{\extendeddithersignal}$ is 
\begin{equation}
    \label{eq:example:covariance-variance-calculation}
    \covariance = \lim_{T\to\infty} \frac{1}{T}\int_{0}^T \error{\extendeddithersignal}(\tau) \error{\extendeddithersignal}(\tau)^T d\tau = \begin{bmatrix} Q_{11} & Q_{12} \\ Q_{21} & Q_{22} \end{bmatrix}
\end{equation}
\begin{equation*}
    \label{eq:example:covariance-q11-matrix}
    Q_{11} = \begin{bmatrix} \left(\frac{1}{2} - \frac{4}{\pi^2}\right) & 0 \\ 0 & \frac{1}{2} \end{bmatrix},\quad Q_{12} = Q_{21}^T = \begin{bmatrix} \frac{1}{3\pi} & 0 & -\frac{1}{2\pi} \\ 0 & \frac{2}{3\pi} & 0\end{bmatrix}
\end{equation*}
\begin{equation*}
    \label{eq:example:covariance-q22-matrix}
    Q_{22} = \frac{1}{8} \begin{bmatrix}
        1 & 0 & -1 \\ 0 & 1 & 0 \\ -1 & 0 & 1
    \end{bmatrix}
\end{equation*}
The matrix $Q_{11}$ is invertible as it is a diagonal matrix with strictly nonnegative numbers. This invertibility means that we can find an $\demodulationsignal$ to estimate the relative gradient. Thus, we can apply gradient-based control laws to the vehicle. In contrast, the second-order perturbation covariance component $\covariance_{22}$ is not invertible. This is a result of the second-order perturbations not being linearly independent. Consequently, we are unable to incorporate second-order derivative information into the control laws. Hence, one such example signal for estimating the relative gradient is
\begin{equation}
    \label{eq:example:covariance-final}
    \demodulationsignal_{\covariance}(t,\ditheramplitude) = \begin{bmatrix}
        2\cos\left(\microphoneangle(t)\right)/\ditheramplitude\left(1 - \frac{8}{\pi^2}\right) \\
        2\sin\left(\microphoneangle(t)\right)/\ditheramplitude
    \end{bmatrix}
\end{equation}   
Although we can formulate $\demodulationsignal_{\covariance}$ to estimate the relative gradient, it is important to note that $\demodulationsignal_{\covariance}$ is not the only signal vector that can be used to estimate the gradient. We can decouple the demodulation signal from $\dithersignal$ by using a different signal for estimating the derivatives. One example is $r$ defined as
\begin{equation}
    \label{eq:example:crossvariance-dithers}
    r(t) = \begin{bmatrix}
        -\cos{(2\pi t)}\\
        -\cos(\pi t)
    \end{bmatrix}
\end{equation}
with the cross-variance matrix
\begin{equation}
    \label{eq:example:covariance-modification}
    R = \frac{1}{2}\int_{0}^2r(\tau)\Tilde{\rho}(\tau)^Td\tau = \begin{bmatrix}
        \frac{2}{3\pi} & 0\\
        0 & \frac{1}{2}
    \end{bmatrix}
\end{equation}  
as an invertible matrix. Therefore, we have found another signal
\begin{equation}
    \label{eq:example:crossvariance-final}
    \demodulationsignal_{\crossvariance}(t,\ditheramplitude) = \begin{bmatrix}
        (-3\pi/2a)\cos\left(2\pi t\right) \\
        (-2/a)\cos\left(\pi t\right) \\
    \end{bmatrix}
\end{equation}
which utilizes trigonometric functions rather than the piecewise functions of the sensor perturbations.

\subsection{Design of Controller}
    \label{sec:nonholonomic-vehicle:controller}

The proposed control law as shown in Fig.\ \ref{blk:example:rotating-microphone:block-diagram} is
\begin{align}
   \label{eq:simulations:seeker-x-y-position}
    \begin{bmatrix}
        \commanded{\forwardvelocity}(t)\\
        \commanded{\angularvelocity}(t)
    \end{bmatrix} &= \begin{bmatrix}
        \gain{\forwardvelocity} & 0 \\
        0 & \gain{\angularvelocity}
    \end{bmatrix}\demodulationsignal(t,a) \costfunction\left(r_s\right)
\end{align}
where
\begin{equation*}
    r_s = r_c + a\begin{bmatrix}
            \cos(\headingangle) & -\sin(\headingangle) \\ \sin(\headingangle) & \cos(\headingangle)
        \end{bmatrix}\dithersignal(\omega t)
\end{equation*}
 is the position of the sensor, $r_c = [x_c,\ y_c]^T$ is the center of the vehicle, and $\demodulationsignal$ can be either $\demodulationsignal_{\covariance}$ or $\demodulationsignal_{\crossvariance}$, and $r_0$ is a small parameter to avoid singularities. The reasoning behind this control law is simple; the vehicle should align itself with the gradient and move forward at a rate proportional to the projection of the gradient on the vehicle's forward axis. By moving the vehicle in a direction aligned with the gradient, $\costfunction$ increases, implying that the vehicle has moved closer towards the source. Showing the practical stability with this controller is beyond the scope of this work; we only wish to heuristically demonstrate the effects of different choices of $\demodulationsignal$ on the trajectories, as well as to show that designating local sensor exploration to a servo yields a better convergence speed of the vehicle towards the source.

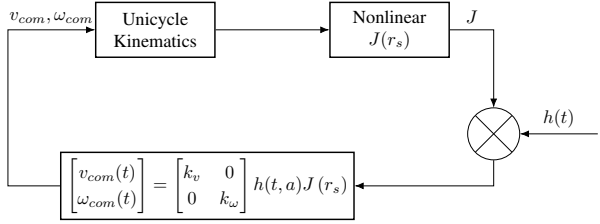
\begin{figure}[t!]
    \centering
    \resizebox{0.45\textwidth}{!}{
    \begin{tikzpicture}[auto, node distance=2cm,>=Latex]
    \node[text width=13ex,align= center] at (-0.5,0)[block] (vehicle) {Unicycle Kinematics};
    \node[text width=13ex,align=center] at (4,0)[block] (cost) {Nonlinear $\costfunction(r_s)$};
    \node at (6,-2) [sum] (convolution) {};
    \node at (0.5,-3) [block] (control-law) {$   
            \begin{bmatrix}
        \commanded{\forwardvelocity}(t)\\
        \commanded{\angularvelocity}(t)
    \end{bmatrix} = \begin{bmatrix}
        \gain{\forwardvelocity} & 0\\
        0 & \gain{\angularvelocity}
    \end{bmatrix}\demodulationsignal(t,a) \costfunction\left(r_s\right)$};
    \node [input,right of =convolution] (demodulation) {};     
    \draw (6.35,-1.65) -- (5.65,-2.35);
    \draw (6.35,-2.35) --(5.65,-1.65);
    \draw[->] (vehicle) -- (cost);
    \draw[->] (cost.east) -| node [pos=0.25,above] {$\costfunction$}(convolution.north);
    \draw[->] (demodulation) -- node [midway,above] {$\demodulationsignal(t)$}(convolution.east);
    \draw[->] (convolution) |- (control-law);
    \draw[->] (control-law.west) -- ++ (-1,0) |- (vehicle.west) node[pos=0.75,anchor=south]{$\commanded{\forwardvelocity}, \commanded{\angularvelocity}$};
\end{tikzpicture}
    }
    \caption{Block diagram for the nonholonomic unicycle gradient-based ESC}
    \label{blk:example:rotating-microphone:block-diagram}
\end{figure}

\begin{subsection}{Simulation With an Acoustic Source}

For the simulation results, we consider a simulated acoustic monopole as the source. This gives our sensor model as
\begin{multline}
    \label{eq:example:acoustic-spoof}
    \costfunction(r_s) = -20\log_{10}\left(\max\left\lbrace\left\Vert r_s - r^*\right\Vert, r_0\right\rbrace\right)  \\ + 20\log_{10}\left(S/\reference{\pressure}\right)
\end{multline}
where $\reference{\pressure} = 20$ $\mu$Pa is the reference pressure, $r^*$ is the position of the source, $S$ is an arbitrary number, and $r_0$ is a small number to avoid singularities. The $\costfunction(r_s)$ in \eqref{eq:example:acoustic-spoof} represents the sound pressure level of the sensor \cite[Eq.~B.5]{ref:salomons-2001}, where we set $S$ so that $\costfunction$ is $80$ decibels when the sensor is one meter away from the source. In Figs.~\ref{fig:demodulation-signal-distance-v-time} and \ref{fig:simulation:rotating-microphone-results-acoustic-source-math-sim}, the results show that the seeker using $\demodulationsignal_{\crossvariance}$, from \eqref{eq:proof:crossvariance-demodulation}, is able to converge to a closer neighborhood of the source compared to the seeker using $\demodulationsignal_{\covariance}$ from \eqref{eq:proof:demodulation-equation}.   

\begin{figure}[t]
    \centering
    \includegraphics[width=1\linewidth]{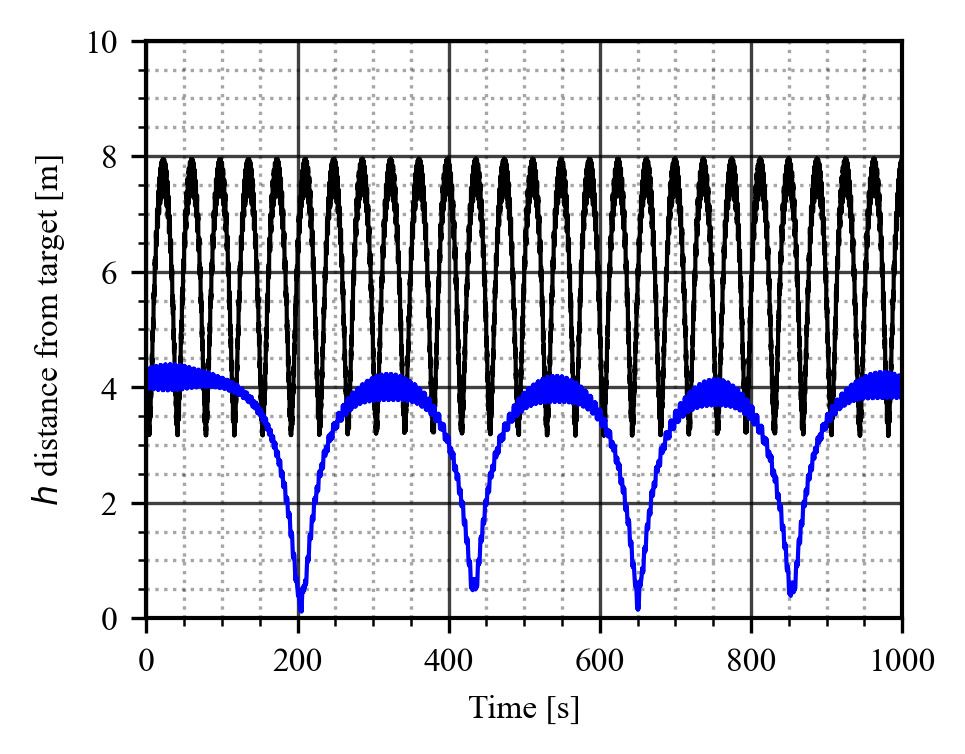}
    \caption{
        Demodulation signal measuring the euclidean distance over time for both the covariance and cross-variance examples.
        \\
        \textbf{KEY}: 
        $\demodulationsignal_{\covariance}=$ \protect\raisebox{1.5pt}{\protect\tikz \protect\draw[black, line width=3pt] (0,0) -- (0.5, 0);} ,
        $\demodulationsignal_{\crossvariance}=$ \protect\raisebox{1.5pt}{\protect\tikz \protect\draw[blue, line width =3pt] (0,0) -- (0.5,0);}.
    }
    \label{fig:demodulation-signal-distance-v-time}
\end{figure}

\begin{figure}[t]
    \centering
    \subfloat[]{
        \includegraphics[width=3.25in]{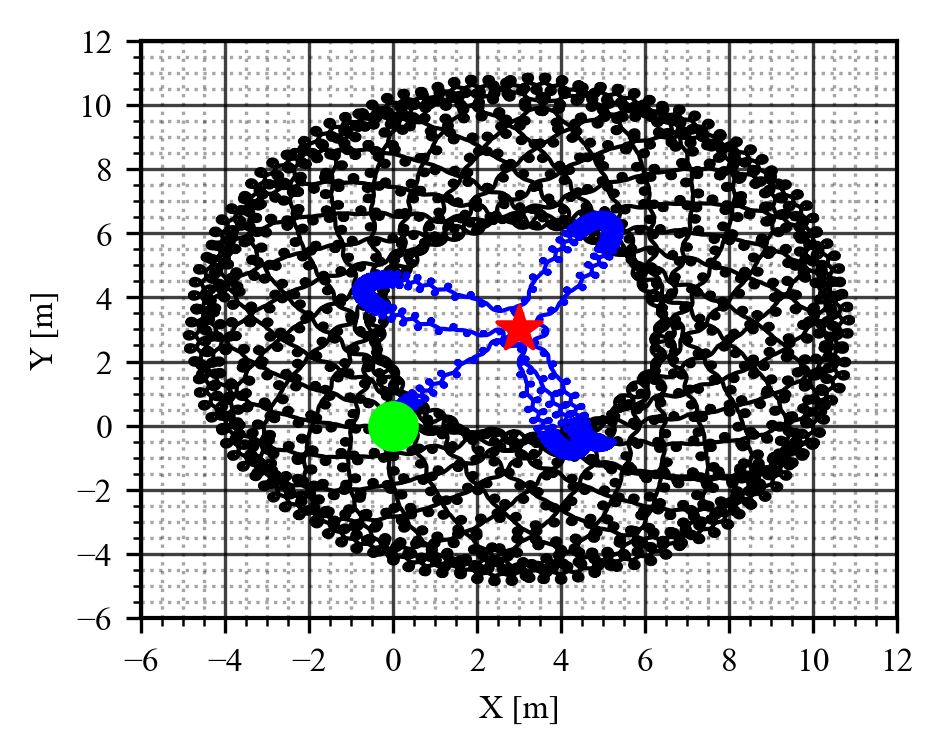}%
        \label{subfig:example:seeker-path-covariance-crossvariance}
    }
    \hfill
    \subfloat[]{
        \includegraphics[width=3.25in]{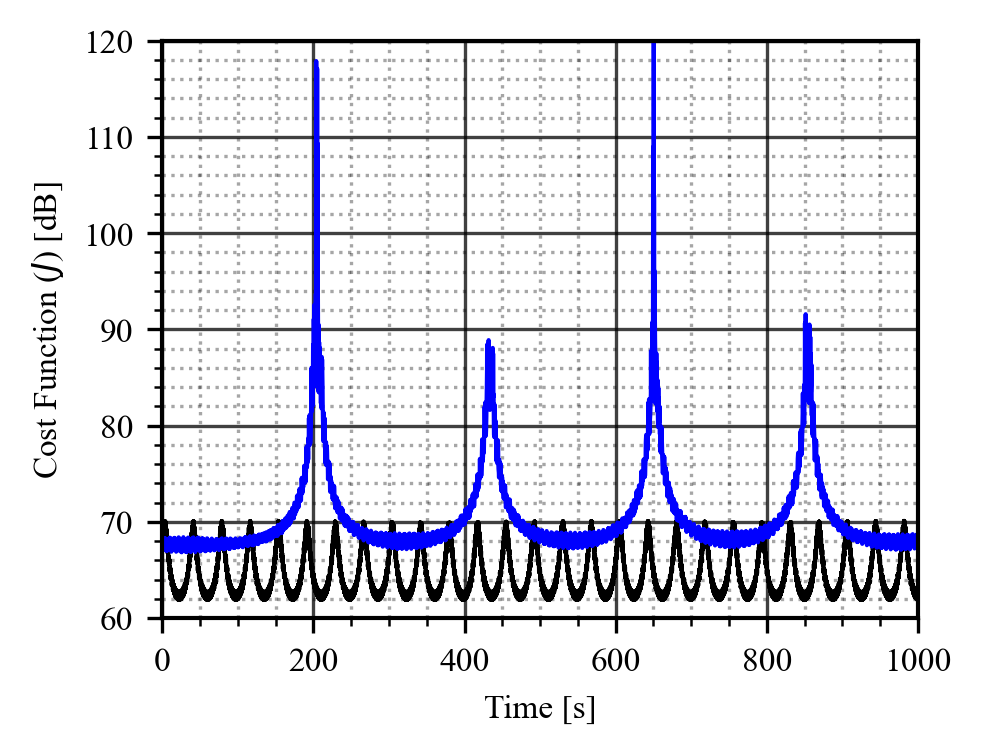}%
        \label{subfig:example:cost-value-hq-hr}
    }
    \caption{
    Simulation results of the seekers with the acoustic source model. The trajectories show the results of using $\demodulationsignal_\covariance(t)$ and $\demodulationsignal_\crossvariance(t)$ with the same gains. The source is located at $r^* = (3,\ 3)^T$ and the gains are $\gain{\forwardvelocity} = 2e-4$ and $\gain{\angularvelocity} = 9.5e-3$. The large parameter selected is $\omega = 40$ rpm ($4\pi/3$ rad/s) and the dither amplitude is $\ditheramplitude = 0.154$ meters. Although both $\demodulationsignal$'s result in orbits of the source, the seeker using $\demodulationsignal_{\covariance}$  has an orbit bounded by a larger radius and remains farther away from the source when compared to the seeker using $\demodulationsignal_{\crossvariance}$.\\
    \textbf{KEY}: 
    $\demodulationsignal_\covariance =$ \protect\raisebox{1.5pt}{\protect\tikz \protect\draw[black, line width=3pt] (0,0) -- (0.5, 0);},
    $\demodulationsignal_\crossvariance =$ \protect\raisebox{1.5pt}{\protect\tikz \protect\draw[blue, line width =3pt] (0,0) -- (0.5,0);},
    Seeker initial position $=$ {\protect\tikz \protect\filldraw[green] (0,0) circle (4pt);},
    Source $=$ {\protect\tikz \protect\node[star, 
              star points=5, 
              star point height=0.225cm, 
              fill=red,
              draw,
              red,
              scale=0.5] {};
    }.
    }
    \label{fig:simulation:rotating-microphone-results-acoustic-source-math-sim}
\end{figure}
\end{subsection}

\subsection{Experimental Results Using a Light Source}

We performed experimental tests in both the Gazebo simulation toolbox and in a real-world laboratory setting using both the controller we developed earlier and several from the previous literature. The simulator was essential for verifying the effectiveness of the method and for selecting the proper controller parameters that do not lead to violations of the vehicle's wheel speed constraints. Once these parameters were established we implemented the control algorithm on the actual robot.

Instead of using an acoustic source, as in previous simulations (in Fig. \ref{fig:simulation:rotating-microphone-results-acoustic-source-math-sim}), we utilized a light source here. We used a photoresistor as the sensor and the cost function as the resistance across the photoresistor. For an ideal photoresistor, the resistance drops to zero as the light intensity increases and there is no directional effects, i.e.,\ the resistance across the photoresistor only depends on the light intensity and not the orientation of the photoresistor. However, for a real photoresistor, there are changes in resistance depending on the relative orientation of the sensor to the source, i.e.,\ the resistance is less when the photoresistor is pointed at the source than when it is pointing away from the source. We denote the angle between the forward direction of the photoresistor and the source by $\sensorangletosource$ (measured in degrees) and use the quadratic cost function
\begin{multline}
    \label{eq:experiments:quadratic-curve-fit}
    \costfunction(r_{s} - r^*,\sensorangletosource) = w_1\Vert r_{s} - r^*\Vert^2 + w_2\sensorangletosource^2 + w_3\Vert r_{s} - r^*\Vert\sensorangletosource \\ + w_4\Vert r_{s} - r^*\Vert + w_5\sensorangletosource + w_6
\end{multline}
with the weights $w_1,\ldots,w_6$ given in Table \ref{tab:experiments:coefficients-quadratic-curve} to represent the resistance across the photoresistor. The quadratic $\costfunction$ with the weights was found to give a reasonable fit of the measured resistance for the distance and orientations of the photoresistor relative to the light source. While $\costfunction$ is no longer just a function of the relative position between the source and the vehicle as was the case in the theoretical portion of this work, we have found that ESCs can still converge to the source. Lastly, note that unlike the acoustic source, the source seeking using resistance as a cost function is a minimization problem with controllers trying to minimize $\costfunction$ to move closer to the source.

\begin{table}[bt]
    \centering
    \caption{Coefficient values for \eqref{eq:experiments:quadratic-curve-fit}.}
    \label{tab:experiments:coefficients-quadratic-curve}
    \begin{tabular}{|c|c|}
    \hline
       \textbf{Coefficients}  & \textbf{Value} \\
       \hline
        $w_1$ & $8.28082113\times10^3$\\
        \hline
        $w_2$ & $7.90425287$ \\
        \hline
        $w_3$ & $4.42130406\times10^2$ \\
        \hline
        $w_4$ & $5.36416164\times10^{-13}$ \\
        \hline
        $w_5$ & $1.68542637\times10^{-16}$ \\
        \hline
        $w_6$ & $2.18515692\times10^{-18}$\\
        \hline
    \end{tabular}
\end{table}

After determining the cost function for the Gazebo simulation, we conducted several experiments on light source seeking. We compared the controllers from previous literature to our own, which rotates the sensor with a servo and uses $\demodulationsignal_{\covariance}$ to estimate the derivatives. The vehicle used was a TurtleBot3 Burger from ROBOTIS Robotics which has a maximum speed of approximately 0.2 meters per second. We tested 3 gradient-based extremum seekers from previous literature: a Lie Bracket seeker \cite{ref:stankovic-2011}; a forward velocity tuning seeker \cite{ref:cochran-2009-2}; and an angular velocity tuning seeker \cite{ref:zhang-2007}. For all controllers used, Table~\ref{tab:methods-testing-parameters} lists the controller parameters that successfully drove the vehicle to the source across a wide range of starting positions and orientations. The initial distance of the vehicle from the source ranged from $\approx$1.5 meters to $\approx$3.5 meters while the vehicle's initial orientation varied from pointing directly at the source to pointing directly away from the source. We only display the results for when the source is at $(2.5, 2.5)$ meters, i.e.\ roughly $3.5$ meters from the vehicle, as this was the hardest distance for convergence but the controllers were tuned over a variety of distances.

\begin{table*}[htb]
    \centering
    \caption{General parameters for the methods used for both simulation and laboratory testing}
    \label{tab:methods-testing-parameters}
    \begin{tabular}{|cS[table-format=1.1]|cS[table-format=1.1]|cS[table-format=1.1]|cS[table-format=1.1]|}
        \hline
        \multicolumn{2}{|c|}{\textbf{Lie Bracket}\cite{ref:stankovic-2011}} & 
        \multicolumn{2}{c|}{\textbf{Angular Tuning} \cite{ref:cochran-2009-2}} & 
        \multicolumn{2}{c|}{\textbf{Forward Tuning}\cite{ref:zhang-2007}} & 
        \multicolumn{2}{c|}{\textbf{Rotating Sensor}} 
        \\
        \hline
        \multicolumn{1}{|c|}{symbol} & \multicolumn{1}{c|}{value} & 
        \multicolumn{1}{c|}{symbol} & \multicolumn{1}{c|}{value} & 
        \multicolumn{1}{c|}{symbol} & \multicolumn{1}{c|}{value} & 
        \multicolumn{1}{c|}{symbol} & \multicolumn{1}{c|}{value}
        \\
        \hline
        \multicolumn{1}{|c|}{$\ditherfrequency$} & 0.18 &
        \multicolumn{1}{c|}{$\ditherfrequency$} & 3.0 &
        \multicolumn{1}{c|}{$\ditherfrequency$} & 1.6 &
        \multicolumn{1}{c|}{$\omega$} & 20.0
        \\
        \hline
        \multicolumn{1}{|c|}{$\commanded{\forwardvelocity}$} & 0.1 &
        \multicolumn{1}{c|}{$\commanded{\forwardvelocity}$} & 0.02 &
        \multicolumn{1}{c|}{$\commanded{\angularvelocity}$} & 0.15 &
        \multicolumn{1}{c|}{$\gain{v}$} & -0.05
        \\
        \hline
        \multicolumn{1}{|c|}{$h$} & 1.0 &
        \multicolumn{1}{c|}{$h$} & 0.05 &
        \multicolumn{1}{c|}{$h$} & 0.5 &
        \multicolumn{1}{c|}{$\gain{\angularvelocity}$} & -0.5
        \\
        \hline
        \multicolumn{1}{|c|}{$\gain{}$} & -6.0 &
        \multicolumn{1}{c|}{$\alpha$} & 0.3 &
        \multicolumn{1}{c|}{$\alpha$} & 0.1 &
        \multicolumn{1}{c|}{$\ditheramplitude$} & 0.154
        \\
        \hline
        \multicolumn{1}{|c|}{$\mu$} & 0.05 &
        \multicolumn{1}{c|}{$c$} & -80.0 &
        \multicolumn{1}{c|}{$c$} & -35.0 &
        \multicolumn{1}{c|}{} &
        \\
        \hline
        \multicolumn{1}{|c|}{} &  &
        \multicolumn{1}{c|}{$d$} & 0.0 &
        \multicolumn{1}{c|}{} & &
        \multicolumn{1}{c|}{} &
        \\
        \hline
    \end{tabular}
\end{table*}

The results of a representative light source seeking experiment conducted in both the Gazebo Simulation Toolbox and a real-world environment are shown in Fig. \ref{fig:results-comparison-plots}
\footnote{Different methods of the GESC light source seeking:\ \url{https://www.youtube.com/watch?v=YtoQkZg6jGk&list=PL5UdoAE0Vc9-6TPsQ9ZcHECyqSNZDug41&pp=gAQB}}.
For the results in Fig. \ref{fig:results-comparison-plots}, all seekers start at the origin with an initial heading angle of $\headingangle_0 = 0$ and the source is located at $r^* = (2.5, 2.5)$ meters. While all the seekers get close to the source, some begin to move away and others appear to settle into a close neighborhood of the source. From a qualitative perspective, the trajectories between the Gazebo simulations and the real-world experiments are fairly similar for the rotating sensor controller and the Lie Bracket controller. There are larger differences between the Gazebo simulations and real-world experiments for the angular velocity tuning controller and even more differences for the forward velocity tuning controller. However, all of these controllers seem to be consistent between the Gazebo simulations and real-world experiments for the time it takes them to get reasonably close to the source. This justifies using $\costfunction$ defined in \eqref{eq:experiments:quadratic-curve-fit} as it yields reasonable accurate simulations compared to real-world experiments.

To compare the controllers, we use reach time as the performance criterion to compare the seekers. The reach time is defined in this work as the first time the vehicle is within 0.5 meters of the source or collides with cabinets in the laboratory close to the source. We give reach times for various starting heading angles, $\headingangle_0$, in Table \ref{tab:method-convergence-performance-rate} with the vehicle again starting at the origin and the source at $r^* = (2.5, 2.5)$. Whenever all the controllers successfully drive the vehicle towards the source, our controller using a rotating sensor outperforms the other tested controllers. Both the angular velocity tuning controller and the rotating microphone controller tended to generate trajectories which bent towards the source, which helps to reduce the reach time since this is a more direct path to the source than those taken by the forward velocity tuning and Lie Bracket controllers. However, as the vehicle's initial heading deviates further from the source, the path required becomes more curved, leading to a longer reach time. Eventually the paths are required to bend too much and there are some initial heading angles in which the seeker fails to reach the source. This happened when the vehicle's forward axis is at an angle greater than $90^{\circ}$ relative to the source, e.g. $\headingangle_0=-90^{\circ}$ and $\headingangle_0=-135^{\circ}$. In these cases, the seeker would either fail to turn in the Gazebo simulations or collide with a wall in the real-world room. In contrast, since they are always turning the vehicle, both the forward velocity tuning controller and the Lie Bracket controller consistently reached the source, albeit at larger reach times because the vehicle had to perform the local exploration of $\costfunction$ for derivative estimates. Although failure to converge is likely unavoidable in the angular velocity tuning controller, the rotating sensor controller had a constraint on the servo arm angle. If this constraint were removed, then it is likely that the rotating sensor seeker would converge for these test cases and still maintain its faster reach times.

\begin{table*}[htb]
    \centering
    \caption{Reach Time in Seconds for Various Initial Heading Angles}
    \label{tab:method-convergence-performance-rate}
    \begin{tabular}{|l|cS||cS||cS||cS||cS||}
    \hline
    \multirow{2}{*}{\diagbox[innerwidth = 2.8cm, height = 6ex]{\textbf{Method}}{\textbf{Angle ($\headingangle_0$)}} } & 
    \multicolumn{2}{c||}{\textbf{45°}} & 
    \multicolumn{2}{c||}{\textbf{0°}} & 
    \multicolumn{2}{c||}{\textbf{-45°}} & 
    \multicolumn{2}{c||}{\textbf{-90°}} & 
    \multicolumn{2}{c||}{\textbf{-135°}}           
    \\ 
    & \multicolumn{1}{c}{Gazebo} & \multicolumn{1}{c||}{Lab} &
    \multicolumn{1}{c}{Gazebo} & \multicolumn{1}{c||}{Lab} &
    \multicolumn{1}{c}{Gazebo} & \multicolumn{1}{c||}{Lab} &
    \multicolumn{1}{c}{Gazebo} & \multicolumn{1}{c||}{Lab} &
    \multicolumn{1}{c}{Gazebo} & \multicolumn{1}{c||}{Lab}                         
    \\ 
    \hline
    Forward Tuning   &
    \multicolumn{1}{S|}{915}    & 363 &
    \multicolumn{1}{S|}{950}   & 893 &
    \multicolumn{1}{S|}{975}    & 828 &
    \multicolumn{1}{S|}{1500}   & 883 &
    \multicolumn{1}{S|}{1200}   & 425 
    \\ 
    \hline
    Angular Tuning   &
    \multicolumn{1}{S|}{180}    & 189 &
    \multicolumn{1}{S|}{185}    & 170 &
    \multicolumn{1}{S|}{185}    & 218 &
    \multicolumn{1}{S|}{250}    & \multicolumn{1}{c||}{N/A} &
    \multicolumn{1}{c|}{N/A}      & \multicolumn{1}{c||}{N/A}
    \\ 
    \hline
    Lie Bracket     &
    \multicolumn{1}{S|}{420}    & 332 &
    \multicolumn{1}{S|}{420}    & 343 &
    \multicolumn{1}{S|}{620}    & 300 &
    \multicolumn{1}{S|}{690}    & 302 &
    \multicolumn{1}{S|}{650}    & 487
    \\
    \hline
    Rotating Sensor  &
    \multicolumn{1}{S|}{\bfseries 75}     & \bfseries 70  &
    \multicolumn{1}{S|}{\bfseries 75}     & \bfseries 75  &
    \multicolumn{1}{S|}{\bfseries 82}     & \bfseries 110 &
    \multicolumn{1}{S|}{\bfseries 128}    & \multicolumn{1}{c||}{N/A} &
    \multicolumn{1}{c|}{N/A}      & \multicolumn{1}{c||}{N/A} 
    \\
    \hline
    \end{tabular}
\end{table*}

\begin{figure*}[t!]
    \centering
    \subfloat[]{
        \includegraphics[width=3.25in]{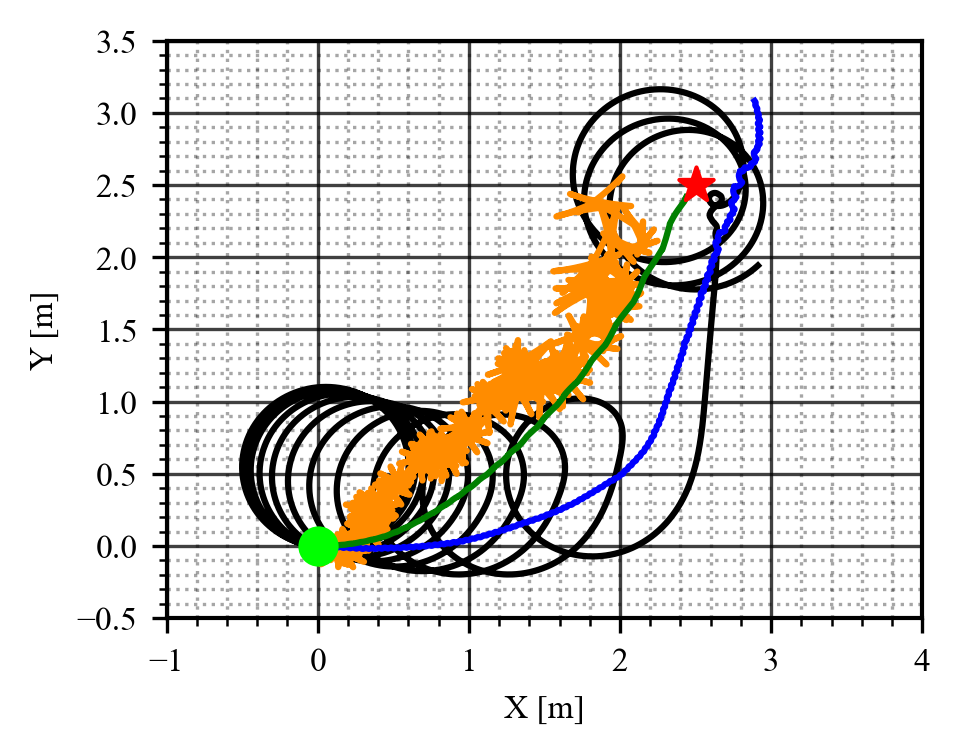}%
        \label{subfig:results-comparison-plots:gazebo-trajectories}
    }
    \hfill
    \subfloat[]{
        \centering
        \includegraphics[width=3.25in]{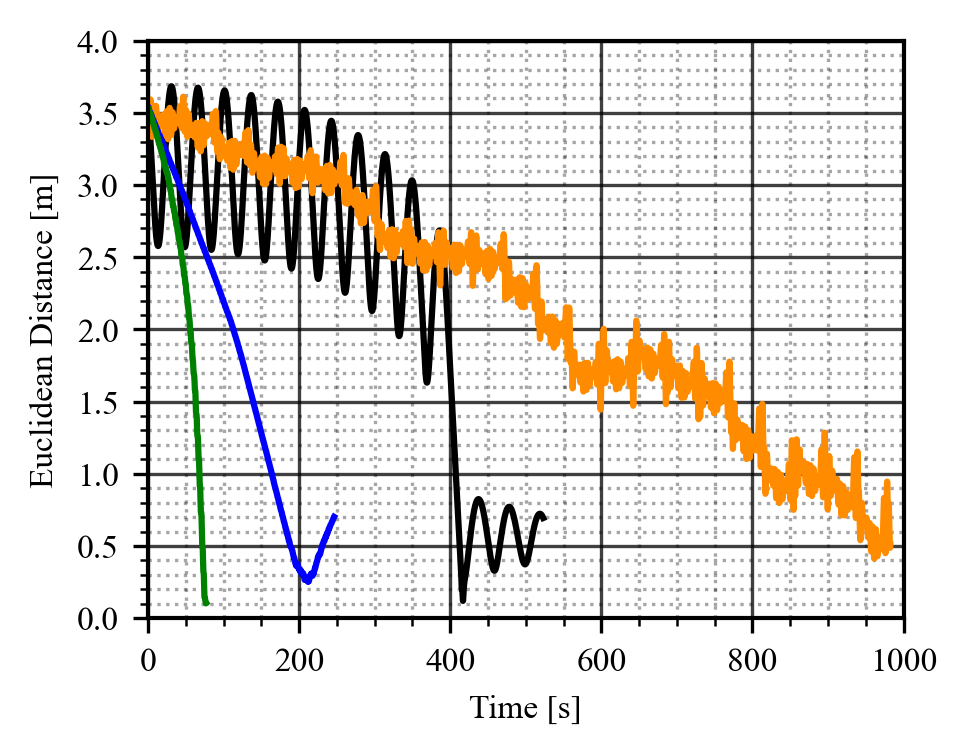}%
        \label{subfig:results-comparison-plots:gazebo-euclidean-distance}
    } \\
    \subfloat[]{
        \includegraphics[width=3.25in]{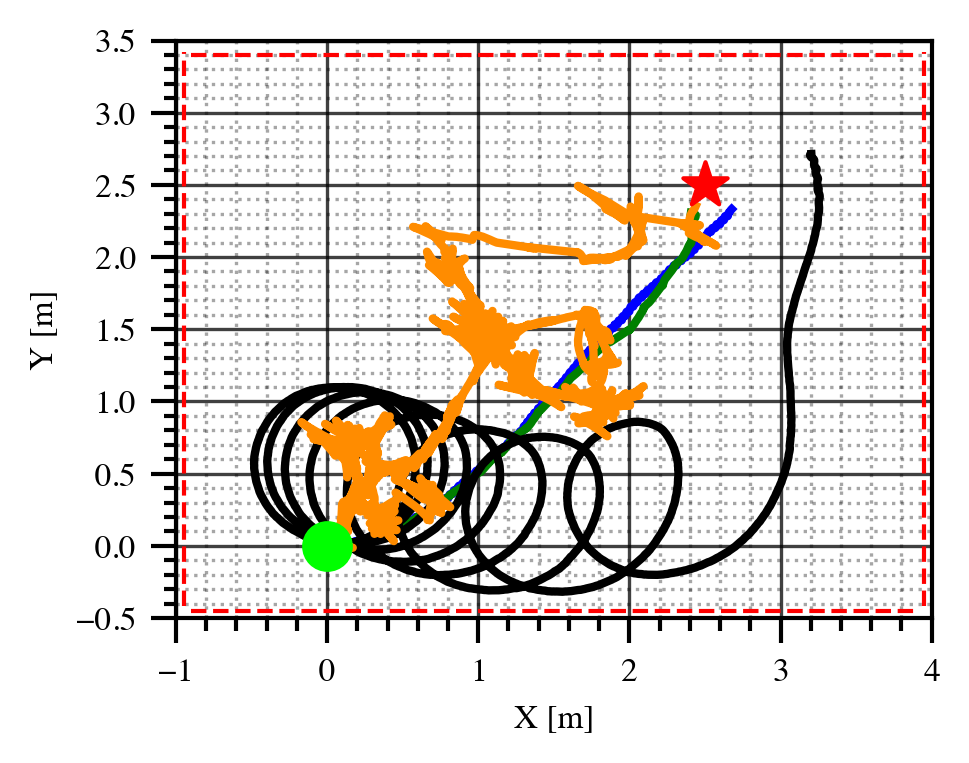}%
        \label{subfig:results-comparison-plots:lab-trajectories}
    }
    \hfill
    \subfloat[]{
        \includegraphics[width=3.25in]{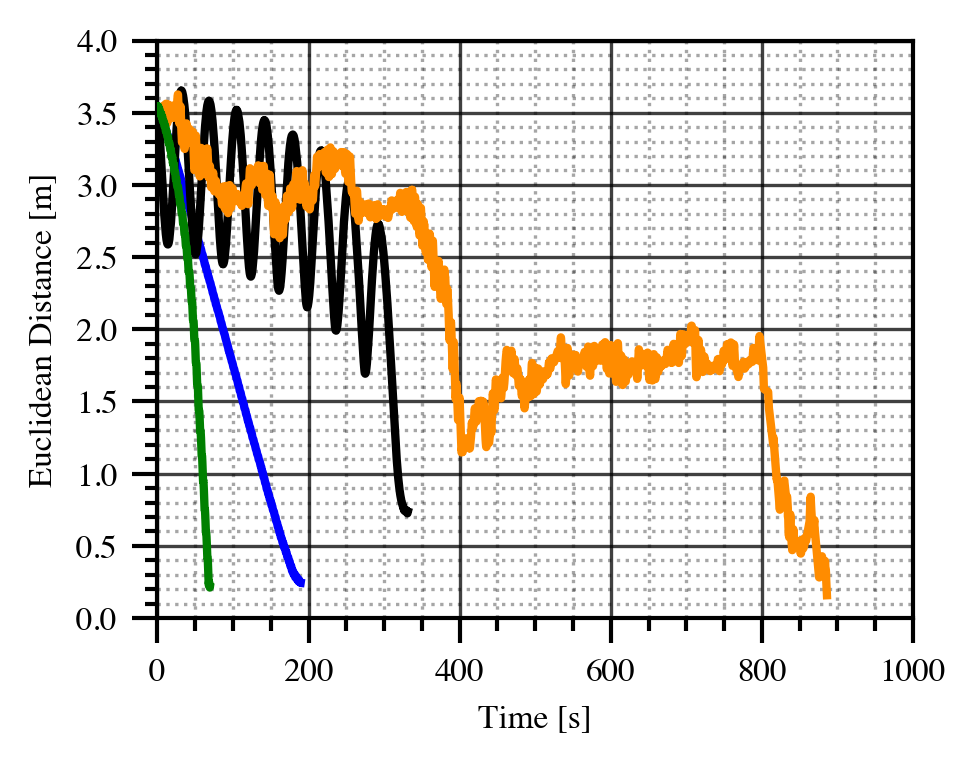}%
        \label{fig:results-comparison-plots:lab-euclidean-distance}
    }
    \caption[Experimental comparison plots]{Experimental results between different ESC designs. The methods being compared are Lie Bracket \cite{ref:stankovic-2011}, angular tuning from \cite{ref:cochran-2009}, forward tuning from  \cite{ref:zhang-2007}, and our method which is the rotating sensor design. Subfigures (a) and (c) are the trajectories over time while (b) and (d) are the Euclidean distance of the seeker from the source over time. Subfigures (a) and (b) are the trajectories in the Gazebo simulation toolbox while Subfigures (c) and (d) correspond to the real-world experiments.
    \textbf{KEY}: 
    Lie Bracket $=$ \protect\raisebox{1.5pt}{\protect\tikz \protect\draw[black, line width=3pt] (0,0) -- (0.5, 0);},
    Angular tuning $=$ \protect\raisebox{1.5pt}{\protect\tikz \protect\draw[blue, line width =3pt] (0,0) -- (0.5,0);},
    Forward tuning $=$ \protect\raisebox{1.5pt}{\protect\tikz \protect\draw[orange, line width =3pt] (0,0) -- (0.5,0);},
    Rotating sensor \protect\raisebox{1.5pt}{\protect\tikz \protect\draw[black!60!green, line width =3pt] (0,0) -- (0.5,0);},
    Room Walls $=$ \protect\raisebox{1.5pt}{\protect\tikz\protect\draw[red,dashed,line width = 3pt](0,0) -- (0.5, 0);},
    Seeker initial position $=$ {\protect\tikz \protect\filldraw[green] (0,0) circle (4pt);},
    Source $=$ {\protect\tikz \protect\node[star, 
              star points=5, 
              star point height=0.225cm, 
              fill=red,
              draw,
              red,
              scale=0.5] {};
    }.
     }
    \label{fig:results-comparison-plots}
\end{figure*}

\section{Conclusions}
    \label{sec:conclusion}
    In this work, we introduce a methodology for estimating derivatives under the assumption of periodic or nearly periodic dither signals, which facilitates the stability analysis of model-free ESCs. Specifically, we establish a necessary and sufficient condition for the existence of time-varying functions which when multiplied by the sensor output result in perturbation-based derivative estimates. We showed that these perturbation-based derivative estimates will converge to the true derivative values after averaging and in the limit of $\ditheramplitude\to0$. 

The usefulness of this result is demonstrated in Section \ref{sec:nonholonomic-ground-vehicle-ex}, where we analyze the existence of a time-varying function $\demodulationsignal$ to estimate the gradient in the relative vehicle frame based on readings from a sensor being actuated by a servo. We present simulations with various $\demodulationsignal$'s to illustrate how the choice of $\demodulationsignal$ leads to different behavior for source seeking controllers. Lastly, we show that a source seeking controller utilizing our results and actuating a sensor on a servo tends to have better performance than the state-of-the-art results in previous literature.
\appendix
\section{Equivalence of Previous Literature}
    \label{sec:appendix}
    To show that the estimation signals of \cite{ref:mills-2018, ref:rusiti-2016} for estimating the $m$-th order derivative of scalar maps can be derived from \cite{ref:nesic-2010-2}, we need to show that following \cite{ref:nesic-2010-2} (and thus the main results of our work) results in 
\begin{multline}
    \demodulationsignal_m(\tau, a) = \frac{2^m m!}{a^m}\left(-1\right)^{\left(\frac{m - \left\vert \sin\left(\frac{m\pi}{2}\right) \right\vert}{2}\right)}\\ \cdot\sin\left(m\tau + \frac{\pi}{4}\left(1 + (-1)^m\right)\right)
\end{multline}
or equivalently
\begin{equation}
    \demodulationsignal_m(\tau, a) = \frac{2^m m!}{a^m}
    \begin{cases}
        \left(-1\right)^{(m-1)/2} \sin\left(m\tau\right) & m\text{ is odd} \\
        \left(-1\right)^{m/2} \cos\left(m\tau\right) & \text{otherwise}
    \end{cases}
\end{equation}
Note that in our work, the scalar sinusoidal perturbation is given as $\dithersignal(\ditherfrequency t) = \sin(\ditherfrequency t)$. Using the exponential form $\sin(\tau) = (e^{j\tau} - e^{-j\tau})/(2 j)$ where $j$ here is $\sqrt{-1}$, we can calculate the extended perturbation signal for the $i$-th order derivative as $\extendeddithersignal_{i}(\tau) = \sin^{i}(\tau)$ for $i=0,1,\ldots,m$ where
\begin{multline}
        \sin^{i}(\tau) = \left(\frac{1}{2j}\right)^i\left\lbrace\left[\sum_{k=0}^{i/2} 2 \binom{i}{k} \left(-1\right)^{i-k} \cos((i-2k)\tau) \right] \right. \\ \left. + \left(-1\right)^{i/2}\binom{i}{i/2}\right\rbrace
\end{multline}
if $i$ is even and 
\begin{multline}
        \sin^{i}(\tau) = \left(\frac{1}{2j}\right)^{i-1} \\ \cdot\left[\sum_{k=0}^{(i-1)/2} \binom{i}{k} \left(-1\right)^{i-k-1} \sin((i-2k)\tau) \right]
\end{multline}
if $i$ is odd. The advantage of these explicit forms for the powers of $\sin$  is that it allows us to immediately see that $\extendeddithersignal(\tau)$ is a vector of linearly independent signals since each $\sin^{i}(\tau)$ contributes a unique $\sin(i \tau)$ or $\cos(i \tau)$ component compared to $\extendeddithersignal_k(\tau)$ associated with lower derivative values ($k < i$). We can determine an orthogonal vector signal $\extendeddithersignal_{\perp}(\tau)$ where
\begin{equation}
    \extendeddithersignal_{\perp,i}(\tau) = \begin{cases}
        1 & \text{ if } i = 0 \\
        \sin(i \tau) & \text{ if } i \text{ is odd} \\
        \cos(i \tau) & \text{otherwise}
    \end{cases}
\end{equation}
using a lower triangular matrix $G$ so that $\extendeddithersignal_{\perp}(\tau) = G\extendeddithersignal(\tau)$ with the diagonal elements of $G$ being
\begin{equation}
    G_{ii} = \begin{cases}
        1 & \text{ if } i=0 \\
        2^{i-1} (-1)^{(i-1)/2} & \text{ if } i \text{ is odd} \\
        2^{i-1} (-1)^{i/2} & \text{ otherwise}
    \end{cases}
\end{equation}
and every other element in the lower diagonal part of $G$ being zero due to the orthogonality of sine and cosine. As none of the diagonal elements are zero, $G$ is invertible. Now consider averaging the perturbation-based derivative estimates for a polynomial
\begin{align}
    \label{eq:equivalence:polynomial-average-derivative estimate}
    \avgest{\derivativevalue}(\avgest{\parameter},\ditheramplitude) &{}={} A^{-1} \covariance^{-1}\left(\lim_{T\to\infty}\frac{1}{T}\int_0^T \extendeddithersignal(\tau) \extendeddithersignal(\tau)^T d\tau\right) A \derivativevalue(\avgest{\parameter}) \\
    &= A^{-1} \covariance^{-1} G^{-1} \covariance_{\perp} G^{-T} A \derivativevalue(\avgest{\parameter})
\end{align}
where 
\begin{align}
    Q_{\perp} &= \lim_{T\to\infty}\frac{1}{T}\int_0^T \extendeddithersignal_{\perp}(\tau) \extendeddithersignal_{\perp}(\tau)^T d\tau \\
    &= \diagonal\left\lbrace 1,\underbrace{\frac{1}{2},\ldots,\frac{1}{2}}_{m\text{-times}}\right\rbrace
\end{align}
is a strictly diagonal matrix. We know from Theorem~\ref{thm:existence-of-consistent-derivative-estimation-signals} that the LHS and RHS of \eqref{eq:equivalence:polynomial-average-derivative estimate} are equal to $\derivativevalue(\avgest{\parameter})$ and so $\covariance^{-1}~=~G^T \covariance_{\perp}^{-1} G$. With this knowledge, reexamining the covariance-based estimation signal yields
\begin{equation}
    \demodulationsignal(\tau,\ditheramplitude) = A^{-1} G^T \covariance_{\perp}^{-1} \extendeddithersignal_{\perp}(\tau)
\end{equation}
with $A_{ii} = \ditheramplitude^{i}/i!$. Since $\demodulationsignal$ is product of diagonal matrices $A^{-1}$ and $\covariance_{\perp}^{-1}$ as well as the upper block diagonal matrix $G^T$, it is easy to verify that
\begin{equation*}
    \demodulationsignal_m(\tau, a) = \frac{2^m m!}{a^m}
    \begin{cases}
        \left(-1\right)^{(m-1)/2} \sin\left(m\tau\right) & m\text{ is odd} \\
        \left(-1\right)^{m/2} \cos\left(m\tau\right) & \text{otherwise}
    \end{cases}
\end{equation*}
which is what we needed to show. Thus, the derivative estimates in \cite{ref:mills-2018, ref:rusiti-2016} for estimating the $m$-th order derivative of scalar maps were derivable from \cite{ref:nesic-2010-2}.



\bibliographystyle{IEEEtran}
\bibliography{references}


\begin{IEEEbiography}[{\includegraphics[width=1in,height=1.25in,clip,keepaspectratio]{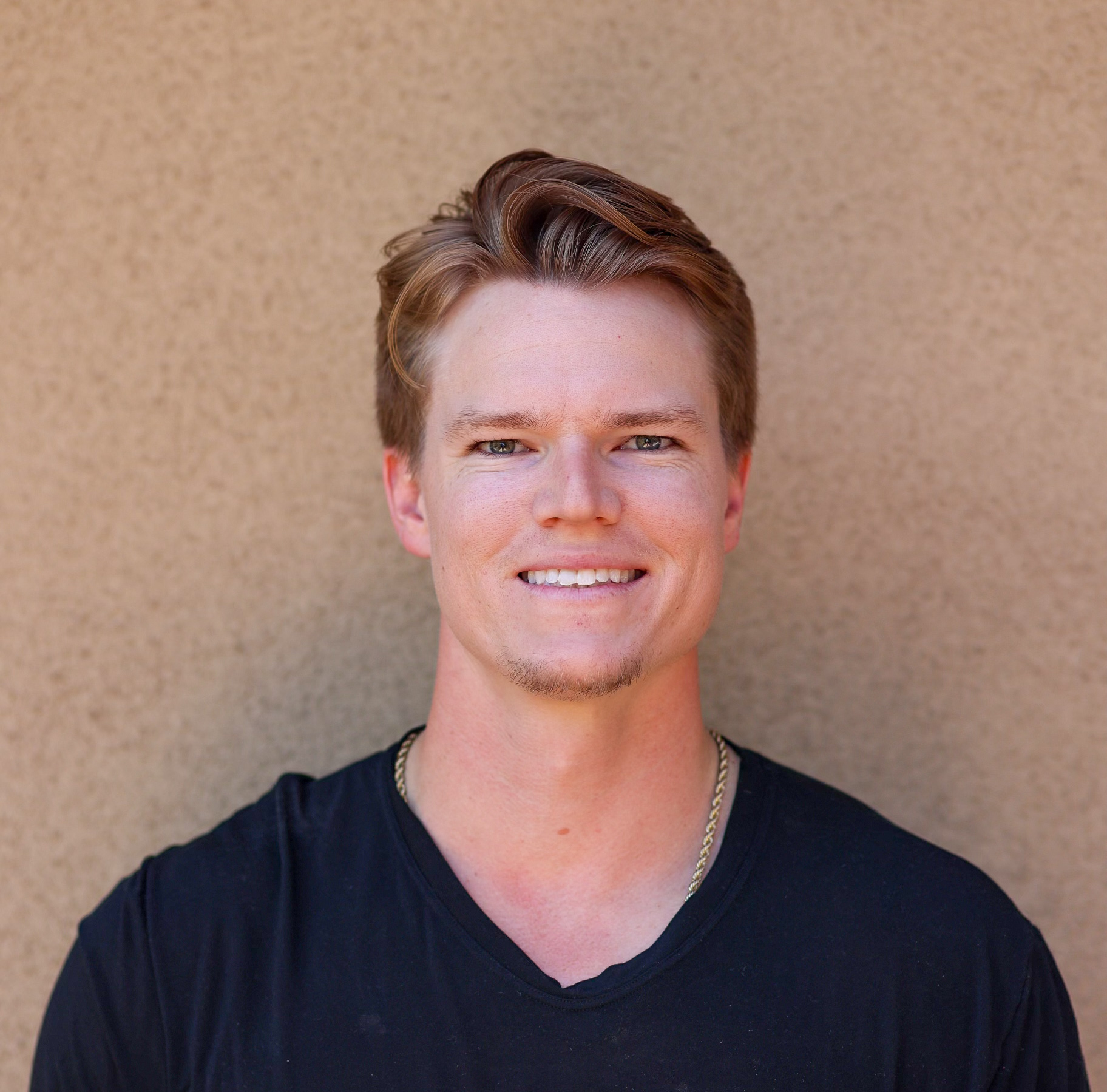}}]{Dylan James-Kavanaugh}
    received the B.S. degree in Mechatronic Engineering from California State University, Chico, Chico, CA, USA, in 2018. Following graduation, he worked at Honeywell Inc. as a field engineer, developing skills in control design for building automation. In 2021, he returned to academia and has been a researcher at the Dynamic Systems and Intelligent Machines Lab (DSIM) since 2023. He is currently obtaining an M.S. degree in Mechanical Engineering from San Diego State University, San Diego, CA, USA, which he is expected to complete in 2025. He is presently focused on pursuing a career in robotics.
\end{IEEEbiography}

\begin{IEEEbiography}[{\includegraphics[width=1in,height=1.25in,clip,keepaspectratio]{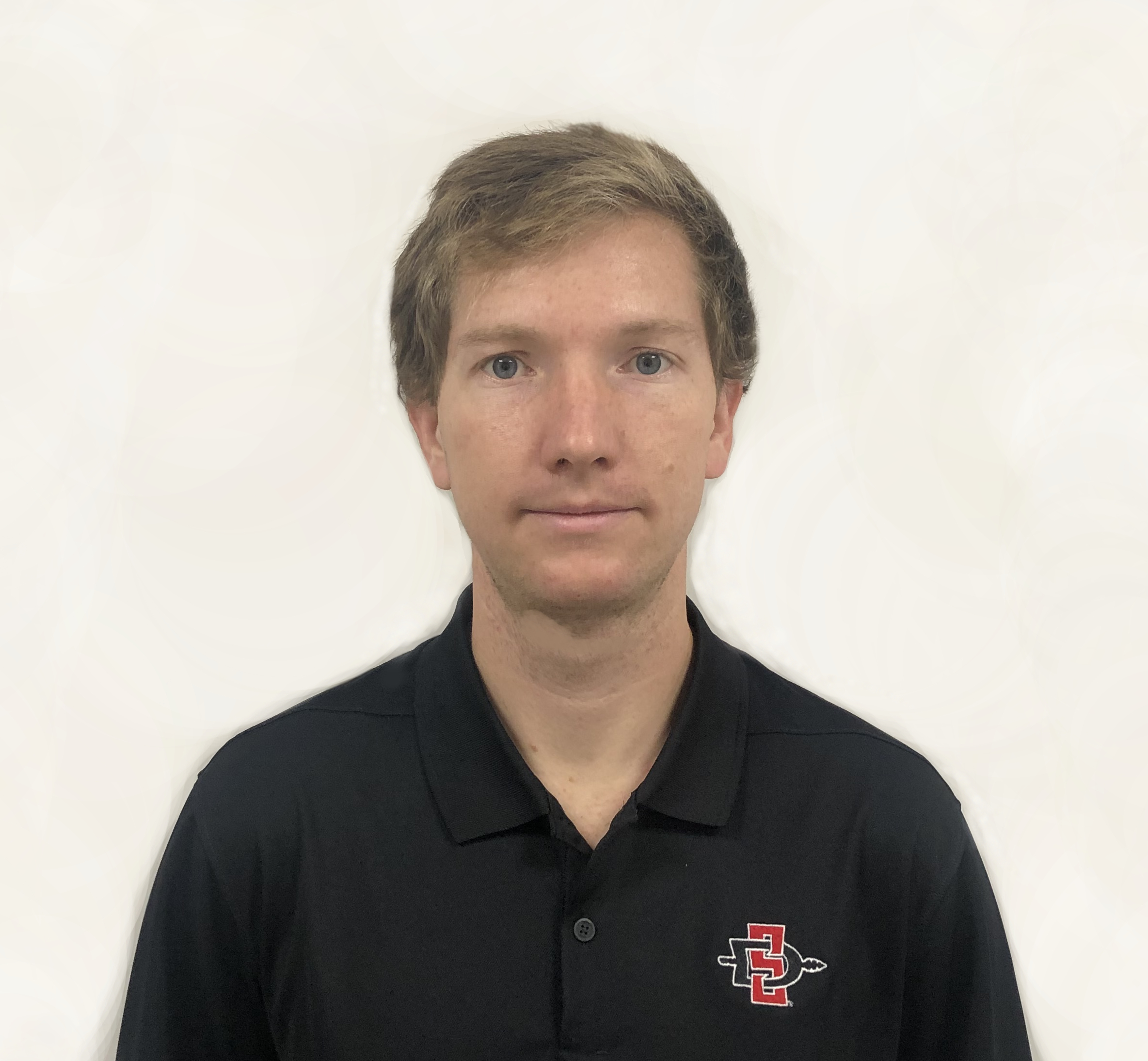}}]{Patrick McNamee}
    received the B.S. degree in aerospace engineering from the University of Kansas, Lawrence, KS, USA, in 2017 before receiving a M.Sc. degree in aerospace engineering and a M.Sc. degree in computer science from the University of Kansas in 2020 and 2021, respectively.
    
    From 2016 to 2021, he was with the University of Kansas as a undergraduate and graduate student research, studying the dynamics and controls of unmanned aerial vehicles (UAVs) as well as machine learning applications for UAVs. He is currently a graduate student researcher with the Dynamic Systems and Intelligent Machines Lab (DSIM) at San Diego State University, having started in 2021.
\end{IEEEbiography}

\begin{IEEEbiography}[{\includegraphics[width=1in,height=1.25in,clip,keepaspectratio]{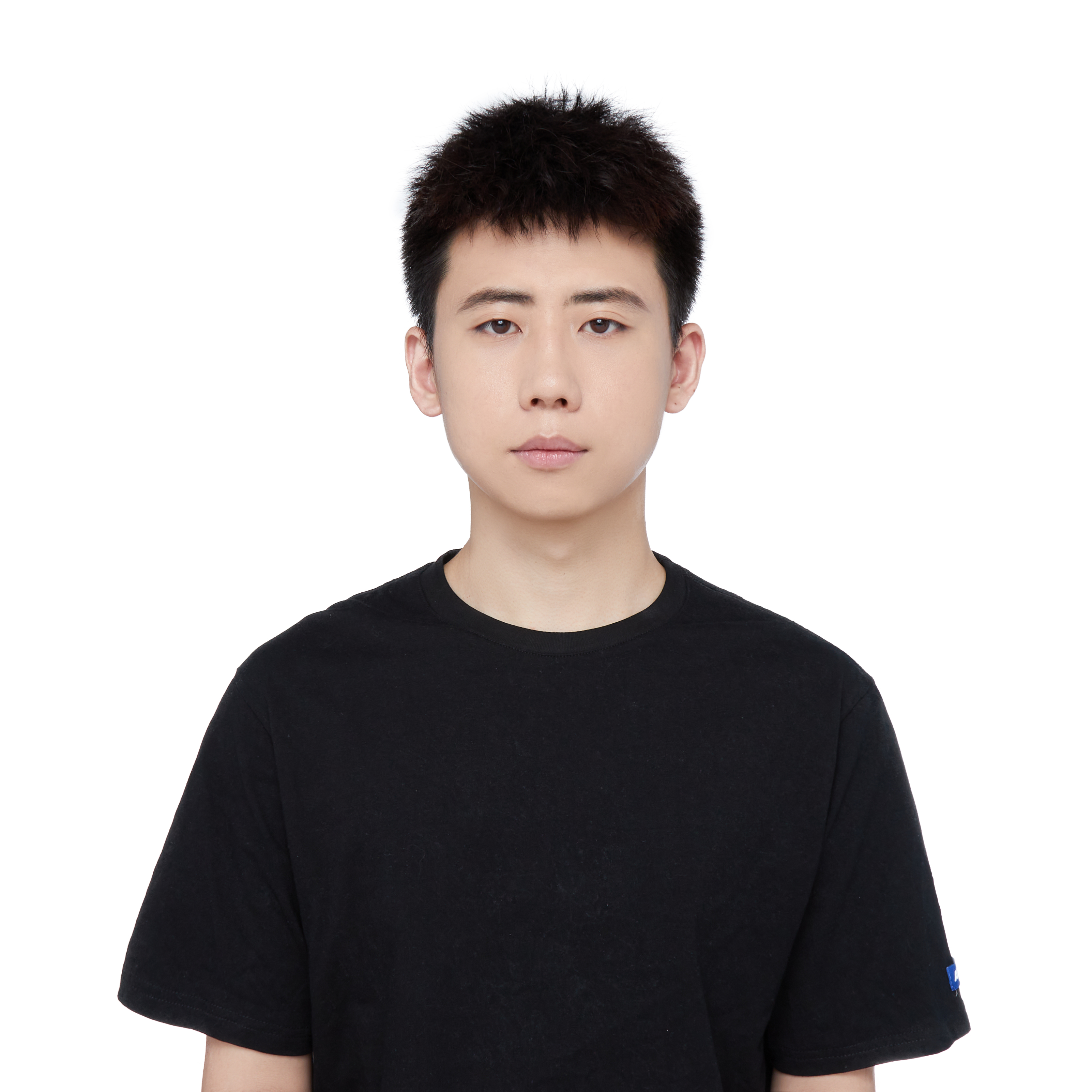}}]{Qixu Wang}
    received the B.Sc. degree in Process Equipment and Control Engineering from the Dalian University of Technology, Dalian, Liaoning, China, in 2019 and  received a M.Sc. degree in Mechanical Engineering from the Northeastern University, Boston, MA, USA, in 2021. He is currently a graduate student researcher with the Dynamic Systems and Intelligent Machines Lab (DSIM) at San Diego State University, having started in 2022.
\end{IEEEbiography}

\begin{IEEEbiography}[{\includegraphics[width=1in,height=1.25in,clip,keepaspectratio]{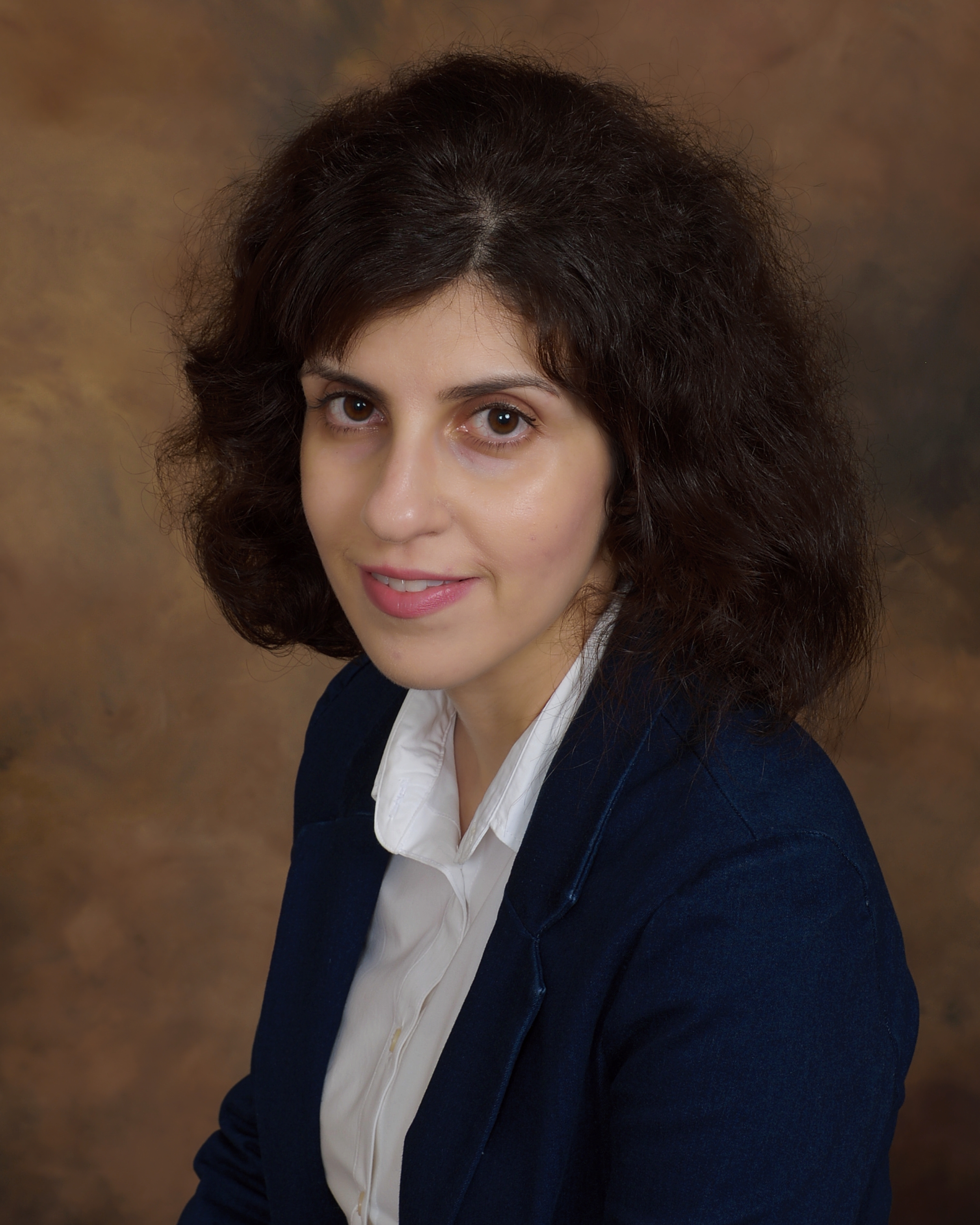}}]{Dr. Zahra Nili Ahmadabadi}
    Zahra Nili Ahmadabadi is an assistant professor in the Mechanical Engineering Department at San Diego State University (SDSU). She is a recipient of the ASME rising star award and ARO Early career award. Nili is currently the associate editor of Sage Journal of Vibration and Control. Her current research interests include robot perception and control, and nonlinear systems.
\end{IEEEbiography}

\end{document}